\documentclass[11pt,letterpaper]{article}

\usepackage{fix-cm}

\usepackage[margin=1in]{geometry}
\usepackage[T1]{fontenc}
\usepackage{times}
\usepackage[fontsize=10.5pt]{fontsize}
\usepackage{authblk}
\usepackage{amsmath}
\usepackage{amssymb}
\usepackage{graphicx}
\usepackage{float}
\usepackage{booktabs}
\usepackage{caption}
\DeclareCaptionFont{captionsize}{\fontsize{9}{11}\selectfont}
\usepackage[super,sort&compress]{natbib}
\usepackage{xurl}
\usepackage[hidelinks]{hyperref}

\renewenvironment{abstract}{%
  \begin{center}%
    {\bfseries \abstractname\vspace{-.5em}}%
  \end{center}%
  \quotation
}{\endquotation}

\title{Reputation and institutional certification as complementary trust mechanisms in a single online market}

\author[1,2,*]{Yuta Kido}
\author[1]{Yohsuke Ohtsubo}

\affil[1]{Graduate School of Humanities and Sociology, The University of Tokyo, Tokyo 113-0033, Japan}
\affil[2]{Japan Society for the Promotion of Science, Tokyo 102-0083, Japan}
\affil[*]{Corresponding author. E-mail: yu-kido@l.u-tokyo.ac.jp}

\date{\today}

\usepackage{placeins}
\usepackage{amsthm}
\usepackage{makecell}
\usepackage{multirow}
\usepackage{enumitem}
\theoremstyle{definition}
\newtheorem{definition}{Definition}[section]
\newtheorem{assumption}{Assumption}
\newtheorem{proposition}{Proposition}[section]
\newtheorem{lemma}{Lemma}[section]
\newtheorem{remark}{Remark}[section]
\newtheorem{corollary}{Corollary}[section]
\newcommand{\SItext}{%
  \section*{Supplementary Notes}\label{sec:si-methods}%
}
\newcommand{\SIresults}{%
  \section*{Supplementary Figures and Tables}\label{sec:si-results-top}%
}
\graphicspath{{figures/}}

\begin{document}

\maketitle

\vspace{-20pt}

\begin{abstract}
Reputation and institutional certification are the two main trust mechanisms under information asymmetry, yet their interaction remains poorly understood. Here we analyze nearly one million eBay listings of Pokémon trading cards. Each listing pairs the seller's reputation with their signal choice, that is, whether they used costly third-party certification, their own costless claim (self-grading), or neither to convey the card's quality. First, we find that certification and self-grading partition the market. Self-grading prevails where reputation is high and value low, whereas certification prevails in the market's opposite corner. Furthermore, seller reputation raises the price premium of self-grading but not of certification. We then show that both empirical patterns follow from a signaling model in which false self-grading damages reputation, whereas certification requires an up-front fee. These results suggest that, even within a single market, trust rests not on a choice between reputation and institutions but on a division of labor.
\end{abstract}

\medskip

\noindent Trust is a precondition for exchange in markets with information asymmetry. When buyers delegate the assessment of product quality to sellers, they are exposed to the risk that sellers will exploit this information advantage\citep{Coleman1990}. Akerlof\cite{Akerlof1970} showed that this asymmetry triggers adverse selection. Unable to verify quality, buyers offer prices below the true value of high-quality products, so high-quality sellers exit the market, leaving only sellers of low-quality products, the ``lemons.'' The problem is amplified in online markets, where buyers can neither inspect the goods directly nor monitor sellers through repeated face-to-face contact. Nevertheless, platforms such as eBay sustain tens of billions of dollars in transactions annually\citep{eBay2025_10K}, implying that alternative trust mechanisms address this information problem.

One such mechanism is reputation. In repeated interactions, reputation sustains cooperation when expected future gains outweigh the short-term benefit of cheating, a principle Axelrod\cite{Axelrod1984} called the ``shadow of the future.'' Online feedback systems embody this principle by making each seller's transaction history available to future buyers\citep{Lehdonvirta2022}. Like brand image\citep{Wernerfelt1988}, an established reputation is then too valuable to lose for one dishonest sale\citep{KrepsWilson1982}. Indeed, reputation scores correlate with price premiums\citep{Resnick2006}, and sellers who receive negative feedback leave the market more often\citep{CabralHortacsu2010}. Thus, reputation systems address information asymmetry by inferring trustworthiness from a seller's accumulated transaction history\citep{ResnickZeckhauser2002, Cook2009}.

A second mechanism operates through third-party institutions, which reduce transaction costs and constrain opportunistic behavior\citep{North1990, Williamson1985}. In online markets, one form this mechanism takes is third-party certification, in which sellers pay an independent service to certify quality (e.g., to prove the authenticity of branded goods\citep{DranovJin2010}). A certificate's credibility rests on the certifier's inspection. Low-quality goods cannot obtain a certificate, so certification separates high-quality goods from low-quality goods regardless of who sells them. Thus, like auditing\citep{Wilson1983}, third-party certification addresses information asymmetry by relocating the source of trust from the seller to an external institution.

\begin{figure}[t!]
\centering
\makebox[\textwidth][c]{\includegraphics[width=18cm]{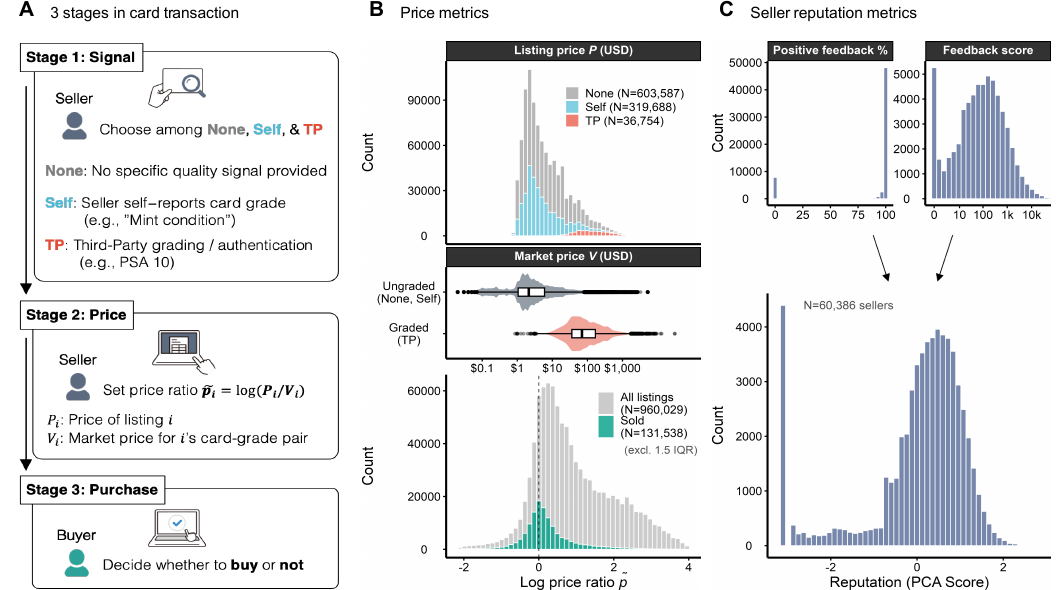}}
\caption{Descriptive overview of the Pokémon card transaction dataset. \textbf{(A)} Three stages of a card transaction. In Stage 1, a seller chooses a quality signal (None, Self-graded, or TP-graded). In Stage 2, the seller sets a listing price relative to the market price. In Stage 3, a buyer decides whether to purchase. \textbf{(B)} Price metrics. Top: distribution of listing prices (log USD) by signal type (None, $n = 603{,}587$; Self, $n = 319{,}688$; TP, $n = 36{,}754$). Middle: market price distributions for ungraded (None, Self) and graded (TP) items shown as horizontal violin-boxplots. Bottom: distribution of the log price ratio $\tilde{p} = \log(P/V)$ for all listings ($N = 960{,}029$) and sold items ($n = 131{,}538$), excluding outliers beyond 1.5 $\times$ the interquartile range (IQR); dashed line marks $\tilde{p} = 0$. \textbf{(C)} Seller reputation metrics. Top left: distributions of positive feedback percentage. Top right: distribution of feedback score (log-transformed). Arrows indicate dimensionality reduction using principal component analysis (PCA). Bottom: composite reputation score from the PCA ($N = 60{,}386$ sellers).}
\label{fig:1}
\end{figure}

These two mechanisms are the principal solutions that the social sciences have proposed for the trust problem in economic and social exchange\citep{CookHardinLevi2005, Yamagishi2011, GreifMokyrTabellini2025}. Trust research classifies them as informal and formal sources of trust and has largely studied them separately or, when comparing them, treated institutions as substitutes for reputation\citep{Zucker1986, Kollock1999, Granovetter1985, Greif1993}. Yet the two often coexist, especially on online platforms\citep{Aoki2001, PavlouGefen2004}. Few studies observe both mechanisms in the same listings, and those that do cover only a few thousand auctions\citep{DewallyEderington2006} or measure certification at the seller level\citep{ElfenbeinFismanMcManus2015}. Thus, how the two mechanisms interact within one market remains understudied\citep{Cook2009}.

Answering this question requires observing an entire market in which every item can carry either mechanism. The market for Pokémon trading cards on eBay meets these requirements. In 2026, a single card in certified top condition sold at auction for \$16.5 million\citep{gwr2026pikachu}, exemplifying the scale of this market. Yet a card's value depends on its condition, which buyers cannot verify from a listing, and counterfeiting is a persistent risk\citep{JinKato2006, PSA2026investment}. Information asymmetry between online sellers and buyers is thus severe. In response, third-party grading has grown into an industry whose largest grader alone grades nearly 20 million cards annually\citep{PSA2026investment}. Of these, Pokémon is the most frequently submitted category\citep{PSA2026pokemon}.

In the eBay Pokémon card market, third-party (TP) grading is the institutional mechanism, whereas reputation rests on a cumulative buyer-feedback system (Fig.~\ref{fig:1}C). Sellers can thus claim high quality in two ways, by having the card TP graded or by describing its condition themselves (self-grading). Alternatively, they can make no claim. We denote these three signaling options TP, Self, and None (Fig.~\ref{fig:1}A). The two signals differ in what they cost the seller: TP grading requires a fee, whereas self-grading costs nothing and is not verified at the time of sale, constituting a form of cheap talk\citep{CrawfordSobel1982}. Nonetheless, such cheap talk signals are widely used\citep{JinKato2006, MayzlinDoverChevalier2014} and have been shown to mitigate information asymmetry in other online markets\citep{Lewis2011}. Their credibility must therefore rest on a source other than the cost of producing them.

Here we propose that an established reputation is what makes such cheap talk credible, and that reputation and certification therefore serve as different trust mechanisms within the same market. We collected 960,029 listings, each carrying the seller's reputation and signal choice, and matched them to the cards' market prices (see Methods). We can therefore examine how sellers choose a signal according to card value and their own reputation, and how buyers respond to that choice (i.e., how the signal choice affects the price premium). We proceed in three steps. First, we address both questions with the full set of listings. Second, to uncover the rationale behind the resulting patterns, we model this market as a signaling game in which false self-grading damages reputation, whereas certification requires an up-front fee. We ask whether the empirical patterns follow from this cost structure. Third, we extend the game to repeated sales and ask how the balance between the two trust mechanisms shifts as a seller's reputation accumulates.

\begin{figure}[t!]
\centering
\makebox[\textwidth][c]{\includegraphics[width=18cm]{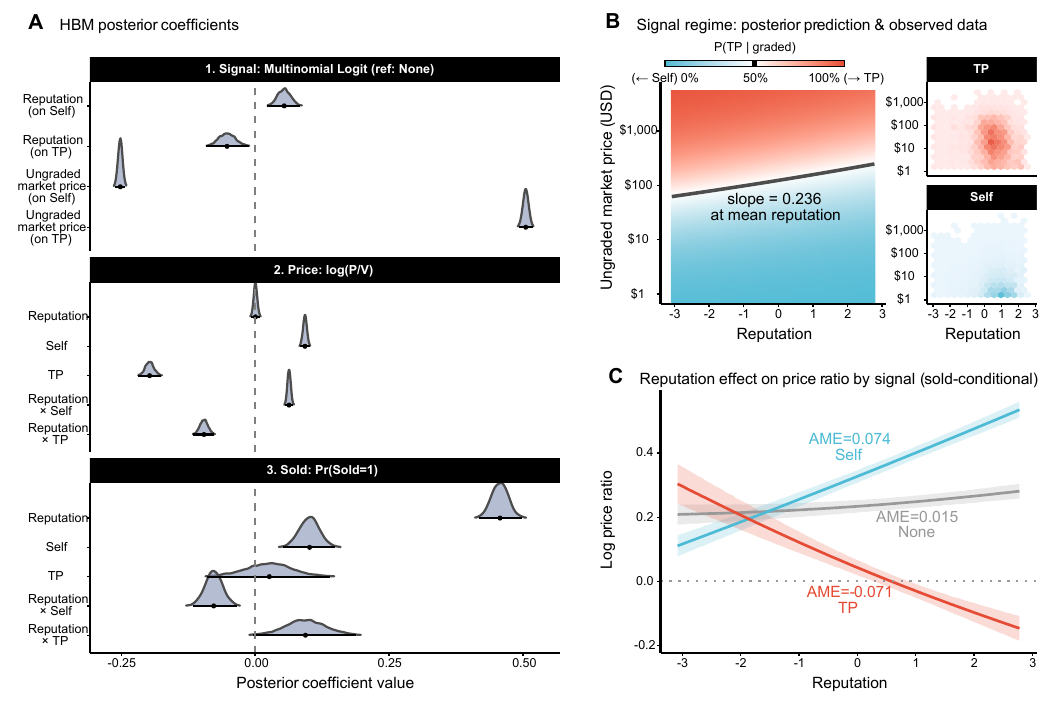}}
\caption{Hierarchical Bayesian model posterior estimates. \textbf{(A)} Ridgeline density plots of posterior coefficient distributions from the three-equation HBM. Top: signal-selection stage (multinomial logit, reference category was None). Middle: price-setting stage (the dependent variable was the log price ratio $\tilde{p} = \log(P/V)$). Bottom: sold stage (the dependent variable was $\Pr(\text{Sold} = 1)$). Vertical dashed line indicates zero. Points within posterior distributions represent medians; horizontal bars show 99\% highest density intervals (HDI). \textbf{(B)} Signal regions: posterior prediction and observed data. Left: region map showing posterior predicted probability of TP grading $P(\text{TP} \mid \text{graded})$ across reputation ($x$) and ungraded market price ($y$, log scale). Red indicates high probability (TP-dominant) and blue low probability (Self-dominant). Boundary contour at 50\%. Right: hexbin density of observed TP-graded (top, red) and Self-graded (bottom, blue) listings. \textbf{(C)} Predicted reputation effect on log price ratio by signal type for sold cards. Lines show predicted marginal effects, with the shaded band giving the 95\% credible interval.}
\label{fig:2}
\end{figure}

\section*{Results}

\subsection*{Empirical analysis}

We fitted a hierarchical Bayesian model (HBM) to $N = 960{,}029$ fixed-price listings, jointly modeling the three sequential stages of an eBay transaction (Fig.~\ref{fig:1}A). In Stage 1, the seller chooses among three signals (TP, Self, or None). Empirically, sellers chose TP in 4\% of listings, Self in 33\%, and None in 63\%. In Stage 2, the seller sets a listing price $P$ relative to the item's grade-appropriate market price $V$, yielding the log price ratio $\tilde{p} = \log(P/V)$ (Fig.~\ref{fig:1}B; see Methods for market price). The denominator $V$ is the graded market price for TP listings and the ungraded market price for Self and None, whereas the item-value predictor in Stage 1 is the ungraded market price for all listings, providing a uniform scale. In Stage 3, a buyer either accepts the listed price or does not. Listing prices skew toward premiums, while sold prices concentrate on the market price (Fig.~\ref{fig:1}B, bottom). The other primary predictor entering all three stages is each seller's reputation score, constructed as the first principal component of two public eBay feedback indicators (feedback score and positive feedback percentage; Fig.~\ref{fig:1}C; see Methods). The HBM converged well for all three stages ($\hat{R} \leq 1.013$ across all 38 parameters). Fig.~\ref{fig:2}A displays the posterior densities for the key parameters (see Supplementary Table~1, Supplementary Fig.~5, and Supplementary Table~2 for convergence diagnostics, model fit, and complete coefficient tables).

First, in the signal-selection stage (Fig.~\ref{fig:2}A, top), we modeled sellers' choices with a multinomial logit. None served as the reference category, and reputation and item value were the primary predictors (see Methods for the full specifications). For reputation, higher-reputation sellers were more likely to choose Self but less likely to choose TP (reputation on Self: $+0.054$ [99\% highest density interval (HDI): $0.023$, $0.083$]; reputation on TP: $-0.052$ [$-0.092$, $-0.011$], see also Fig.~\ref{fig:2}B). For item value, sellers were more likely to choose TP for higher-value items but more likely to choose Self for lower-value items ($\log V$ on TP: $+0.505$ [$0.492$, $0.518$]; $\log V$ on Self: $-0.252$ [$-0.261$, $-0.242$]). As shown in the left panel of Fig.~\ref{fig:2}B, the reputation--value space split into two signal-dominant regions. Third-party grading prevailed in the high-value region, and self-grading dominated in the low-value region. More importantly, the boundary's market price increased with seller reputation. We confirmed this partition with a logistic boundary regression, finding a positive slope (median $\beta = 0.236$ at mean reputation, 99\% HDI: $[0.197, 0.274]$; see Supplementary Fig.~4 for a descriptive boundary with the same slope).

Second, in the price-setting stage (Fig.~\ref{fig:2}A, middle), reputation had a negligible main effect on listing prices ($+0.001$ [99\% HDI: $-0.008$, $0.009$]). Instead, reputation interacted strongly with signal type. The reputation $\times$ self-graded interaction was positive ($+0.063$ [$0.056$, $0.071$]), suggesting that reputation conferred credibility on self-graded listings and generated a corresponding price premium. By contrast, the reputation $\times$ TP grading interaction was negative ($-0.095$ [$-0.115$, $-0.077$]). TP credibility is derived from an external institution rather than the seller, which predicts no reputation effect. Thus, although the absence of a positive effect is consistent with this prediction, the observed negative effect is not readily interpretable. We will discuss the possible interpretations in the Discussion section (see also Supplementary Figs.~2 and 3).

Third, in the purchase stage, the coefficient on the listing-price premium was negative ($-0.951$ [99\% HDI: $-0.969$, $-0.937$]), indicating that higher premiums (i.e., prices above the market) reduced the probability of sale, but reputation counteracted this effect ($+0.457$ [$0.419$, $0.497$]; Fig.~\ref{fig:2}A, bottom). Because only 13.7\% of listings resulted in a sale, listing-price coefficients alone do not capture the premiums actually realized on sold items. To recover these realized premiums while avoiding selection bias from analyzing only sold items, we combined the pricing and purchase models through posterior predictive simulation, yielding sold-conditional expected prices (Fig.~\ref{fig:2}C; see Methods for details). These simulated prices confirmed the reputation $\times$ signal type interaction estimated above. The average marginal effect of reputation on the log price ratio was positive for self-graded listings (AME $= +0.074$) and negative for TP-graded listings (AME $= -0.071$). A descriptive analysis of sold items also corroborated this pattern (Supplementary Fig.~1; see Methods for the simulation procedures).

So far, we have shown that reputation interacts with signal type in two distinct ways. First, the reputation--value space separates along a positively sloping boundary (Fig.~\ref{fig:2}B). Self-grading prevails at high reputation and low value, whereas TP grading prevails at low reputation and high value. Second, reputation amplifies the price premium of costless self-grading while lowering that of costly TP grading. We next develop a signaling game that accounts for these patterns.

\subsection*{Signaling model analysis}

\subsubsection*{Model setup}

We formulate a signaling game in which the seller privately observes product quality and chooses among the three signals (Fig.~\ref{fig:3}; parameters in Table~\ref{tab:1}). Nature draws a quality type $\theta \in \{H, L\}$ with prior probability $P(\theta = H) = \pi_0$. To the buyer, a high-quality item ($\theta = H$) has a value of $V > 0$, while a low-quality item ($\theta = L$) has a value of zero. For the sake of simplicity, we assume that the prior $\pi_0$ is independent of $V$. After observing $\theta$, the seller chooses a signal $s \in \{\text{TP}, \text{Self}, \text{None}\}$. Buyers observe both $s$ and the seller's public reputation $\rho \in [0, 1]$, a cumulative record of past feedback, and form a posterior belief $\hat{p}_s = P(\theta = H \mid s, \rho)$ that the item is high-quality. Under the fixed-price format, in which the seller posts a price, we assume that buyers accept the listed price if it does not exceed the item's expected value $\hat{p}_s V$. Therefore, the equilibrium price for signal $s$ is $\hat{p}_s V$. When the seller sends no signal ($s = \text{None}$), the item trades at the outside-option price $\pi_0 V$, which we take as exogenous and equal to the price under the prior $\pi_0$ (see Methods). We characterize the perfect Bayesian equilibria (PBE) of the model and, where several coexist, select among them with a standard refinement (see Methods).

\begin{table}[t!]
\centering
\caption{Parameters of the theoretical model.}
\label{tab:1}
\fontsize{7}{8.4}\selectfont
\begin{tabular}[t]{lp{4.2cm}l}
\toprule
\fontsize{8}{9.6}\selectfont Symbol & \fontsize{8}{9.6}\selectfont Description & \fontsize{8}{9.6}\selectfont Value / Range\\
\midrule
\addlinespace[0.3em]
\multicolumn{3}{l}{\fontsize{8}{9.6}\selectfont \textbf{Model parameters}}\\
\hspace{1em}$\theta$ & Product quality type & $\{H, L\}$\\
\hspace{1em}$s$ & Seller's signal choice & $\{\text{TP}, \text{Self}, \text{None}\}$\\
\hspace{1em}$\hat{p}_s$ & Buyer's posterior; $P(\theta{=}H \mid s, \rho)$ & $[0,\, 1]$\\
\hspace{1em}$\pi_0$ & Prior $P(\theta{=}H)$ & 0.3\\
\hspace{1em}$\phi$ & Penalty coefficient for detected fraud & 0.5\\
\hspace{1em}$c_{\text{TP}}$ & Cost of third-party authentication & 0.2\\
\addlinespace[0.3em]
\multicolumn{3}{l}{\fontsize{8}{9.6}\selectfont \textbf{State variables}}\\
\hspace{1em}$\rho$ & Seller reputation & $[0,\, 1]$\\
\hspace{1em}$V$ & Item value ($V_L{=}0$ normalized) & $\mathbb{R}_+$\\
\bottomrule
\end{tabular}
\end{table}

\begin{figure}[t!]
\centering
\includegraphics{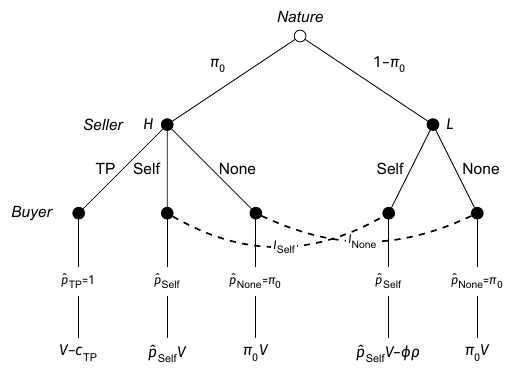}
\caption{Extensive-form game tree of the seller signaling model. Nature assigns quality type $H$ (high quality) with probability $\pi_0$ or $L$ (low quality) with probability $1 - \pi_0$. $H$-type sellers choose among three signals: TP-graded (third-party certification at cost $c_{\text{TP}}$; buyers' posterior $\hat{p}_{\text{TP}} = 1$; the seller's payoff is $V - c_{\text{TP}}$), Self-graded (no cost; reputation unaffected; buyers' posterior $\hat{p}_{\text{Self}}$; payoff $\hat{p}_{\text{Self}} V$), or None (no cost; buyers' posterior $\hat{p}_{\text{None}} = \pi_0$; payoff $\pi_0 V$). $L$-type sellers choose between Self-graded (no cost; reputational damage $\phi\rho$ upon detection; payoff $\hat{p}_{\text{Self}} V - \phi\rho$) and None (no cost; payoff $\pi_0 V$). Dashed curves denote buyer information sets: $I_{\text{Self}}$ includes the two Self-graded nodes, and $I_{\text{None}}$ includes the two None nodes (i.e., based on the signal alone, the buyer cannot distinguish $H$ from $L$). TP-graded is available only to $H$-type sellers, so its node stands alone. The tree shows Stage 1 of Fig.~\ref{fig:1}; the equilibrium price $\hat{p}_s V$ carries Stages 2 and 3.}
\label{fig:3}
\end{figure}

\subsubsection*{Trust through self-claim and reputation}

Consider first the use of self-grading. This signal is available to both high- and low-quality sellers and requires no immediate financial costs. However, these two types face asymmetric consequences under the following assumptions. When an $H$-type seller self-grades, the claim is accurate and the seller's reputation is unaffected. When an $L$-type seller self-grades, the claim is dishonest and triggers buyers' negative feedback after the transaction. This damages the seller's reputation, leading to a future payoff loss of $\phi\rho$, where $\phi > 0$ scales the loss with the seller's current reputation $\rho$.

Due to this asymmetry in the consequences of self-grading, we can derive a condition under which $H$-type sellers choose Self and $L$-type sellers choose None, with incentive-compatible strategies and Bayes-consistent buyer beliefs (i.e., a separating PBE). The $L$-type seller's expected payoff is $\pi_0 V$ when choosing None and $V - \phi\rho$ when choosing Self. Therefore, the $L$-type seller chooses None as long as $\pi_0 V > V - \phi\rho$, which can be rearranged to

\begin{equation}
V < \frac{\phi}{1-\pi_0}\,\rho \equiv \bar{V}(\rho).
\label{eq:vbar}
\end{equation}
As the slope $\phi/(1-\pi_0)$ is strictly positive, $\bar{V}(\rho)$ is an increasing linear function of $\rho$, which appears as the positively sloping line in Fig.~\ref{fig:4}. The blue region below this line corresponds to the Self-separating equilibrium. Intuitively, $\bar{V}(\rho)$ is the maximum item value at which reputation alone keeps self-grading honest. A seller with a higher reputation has more to lose from misrepresentation and can therefore credibly claim high quality for higher-value items. An established reputation thus serves as collateral backing the seller's own claim, and we call $\bar{V}(\rho)$ the reputation-collateral bound.

\subsubsection*{Trust through institutional certification}

When the item value $V$ exceeds $\bar{V}(\rho)$, the deception gain outweighs the reputation collateral, and self-grading supports at most a semi-separating outcome. $L$-type sellers then choose Self dishonestly with positive probability, and a self-claim is only partially credible (see Methods). Even then, buyers can trust third-party certification. Because $L$-type items cannot pass third-party authentication, observing TP perfectly identifies the item as high-quality, and this credibility owes nothing to the seller's reputation. The certifying institution charges a fixed cost $c_{\text{TP}}$, which is borne by $H$-type sellers (see Supplementary Note~7 for value-dependent fees). For the $H$-type, the TP payoff must be at least the None payoff, $V - c_{\text{TP}} \geq \pi_0 V$, which requires

\begin{equation}
V \geq \frac{c_{\text{TP}}}{1-\pi_0} \equiv \underline{V}.
\label{eq:vunderline}
\end{equation}
Because $c_{\text{TP}}$ does not depend on $\rho$, this bound is the same at every reputation level. Reaching it makes certification viable, in the sense that a TP-separating equilibrium exists. Viability alone does not yet make TP grading the $H$-type's preferred signal. Above $\bar{V}(\rho)$, a partially credible self-claim pays the $H$-type $\pi_0 V + \phi\rho$ rather than $\pi_0 V$. The self-claim thus carries a premium of $\phi\rho$, exactly the penalty that a false claim would incur. Certification pays the $H$-type more than the self-claim only when $V - c_{\text{TP}} > \pi_0 V + \phi\rho$, that is, when $V > \bar{V}(\rho) + \underline{V}$ (see Methods for the equilibrium selection). This certification boundary is the reputation-collateral bound shifted upward by $\underline{V}$. The two boundaries in Fig.~\ref{fig:4} are therefore parallel and enclose a band of vertical width $\underline{V}$. Above the certification boundary, $H$-type sellers choose TP, and $L$-type sellers, unable to pass authentication, choose None. Full separation is restored, now resting on an institutional ex-ante cost rather than on a reputational ex-post penalty.

\subsubsection*{Regime structure and empirical correspondence}

\begin{figure}[t!]
\centering
\includegraphics{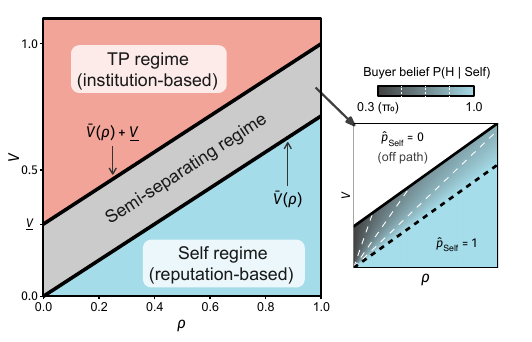}
\caption{Three-regime map of the reputation--value space. The map (left) partitions the $(\rho, V)$ space, where $\rho$ is seller reputation and $V$ is item value, into three regimes, colored categorically. Below the lower boundary $\bar{V}(\rho)$ (Eq.~\ref{eq:vbar}) lies the Self regime (blue), in which self-grading fully separates the types. Above the upper boundary $\bar{V}(\rho) + \underline{V}$ lies the TP regime (red), in which the $H$-type prefers TP grading. Between the two parallel boundaries lies the Semi-separating regime (gray), a band of vertical width $\underline{V}$ (Eq.~\ref{eq:vunderline}). The right panel shows the same space with item value cut at $\bar{V}(1) + \underline{V}$, colored by the buyer's belief in a self-claim, $\hat{p}_{\text{Self}}$. The belief is $1$ throughout the Self regime and rises with reputation as $\pi_0 + \phi\rho/V$ across the Semi-separating band. Above the band, Self is off the equilibrium path and is left blank. Baseline parameters: $\pi_0 = 0.3$, $\phi = 0.5$, $c_{\text{TP}} = 0.2$.}
\label{fig:4}
\end{figure}

\begin{figure}[b!]
\centering
\includegraphics{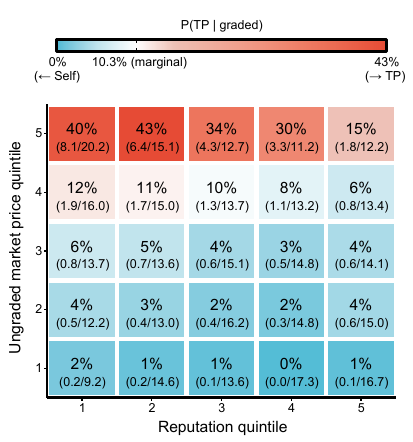}
\caption{Empirical TP-graded certification rates across reputation and market price. Discrete $5 \times 5$ heatmap of the third-party certification ratio $P(\text{TP} \mid \text{graded})$, the share of TP-graded listings among all graded listings, across seller reputation quintiles (columns, $\rho$) and ungraded market price quintiles (rows, $V$). Each cell reports the TP-graded proportion and the sample sizes ($n_{\text{TP}}/n_{\text{total}}$, in thousands, rounded to the nearest hundred). The diverging color scale is centered on the marginal TP rate, with red for above (TP-dominant) and blue for below (Self-dominant).}
\label{fig:5}
\end{figure}

Together, the two boundaries partition the reputation--value space $(\rho, V)$ into three regimes (Fig.~\ref{fig:4}), each corresponding to a different way of producing trust. In the Self regime, reputation alone sustains honest self-claims, and the institution is not used. In the Semi-separating regime, reputation still constrains self-claims but only in part; self-claims are discounted, and no signal fully separates the types. In the TP regime, the institution replaces reputation as the source of credibility.

The theoretical regime structure generates a testable prediction. The certification boundary $\bar{V}(\rho) + \underline{V}$ is positively sloping in $(\rho, V)$ space, as $\bar{V}(\rho)$ increases with reputation. This prediction is consistent with the empirical region map (Fig.~\ref{fig:2}B), whose boundary has the same positive slope. This pattern is further corroborated by the discrete heatmap in Fig.~\ref{fig:5}, which partitions the data into reputation--value quintile cells and shows the TP share among graded listings (see Supplementary Fig.~6 for the marginal decompositions along each axis). Third-party grading peaks in the low-reputation, high-value corner, while self-grading escalates sharply as reputation increases and value decreases.

While the three regimes describe which signal sellers use, the pricing interaction in Fig.~\ref{fig:2}C arises within the Semi-separating regime. The right panel of Fig.~\ref{fig:4} shows, across the regimes, the buyer's posterior belief $\hat{p}_{\text{Self}}$ that a self-graded item is high-quality. In the Self regime, self-grading separates the types fully, so $\hat{p}_{\text{Self}} = 1$, independent of reputation. In the Semi-separating regime, low-quality sellers mix between a dishonest self-claim (payoff $\hat{p}_{\text{Self}}V - \phi\rho$) and None (payoff $\pi_0 V$), and their indifference between the two pins this belief at $\hat{p}_{\text{Self}} = \pi_0 + \phi\rho/V$ (see Methods; Supplementary Note~7). The belief increases in reputation, $\partial\hat{p}_{\text{Self}}/\partial\rho = \phi/V > 0$, because higher reputation partially deters misrepresentation. The equilibrium price equals $\hat{p}_{\text{Self}}\,V$, so the gradient accounts for the positive effect of reputation on the self-grading premium. TP grading, by contrast, derives its credibility from an external institution rather than the seller, so $\hat{p}_{\text{TP}} = 1$ for all $\rho$ and $\partial\hat{p}_{\text{TP}}/\partial\rho = 0$, predicting no reputation $\times$ TP grading interaction (see Discussion for interpretations of the observed negative interaction).

Taken together, the signaling model recovers the equilibrium structure underlying our empirical analysis, reproducing the empirical region map (Fig.~\ref{fig:2}B) and the positive reputation $\times$ self-graded interaction (Fig.~\ref{fig:2}C) from the difference in cost structure between the two trust mechanisms.

\subsection*{Reputation dynamics}

The signaling game above takes reputation as given. We now build on it by letting reputation evolve over repeated sales. Within each period, buyers and sellers play the static equilibrium of the game at the current reputation, with buyers forming the posterior $\hat{p}$ from that reputation. Across periods, reputation changes with the outcomes of trade. We assume that the sale of a high-quality item raises reputation and that a detected misrepresentation lowers it by the same amount, so a single rate $\alpha$ governs both changes. Because quality is drawn per listing, a share $\pi_0$ of a seller's listings is high quality, and, with no misrepresentation, reputation would grow at the rate $\pi_0\alpha$.

Consider a seller for whom certification is unavailable, facing items of value $V$, so that self-grading is the only quality signal. The resulting dynamics have a single interior rest point, which is unstable (Fig.~\ref{fig:6}A; see Methods). Below it, misrepresentation is frequent enough that reputation falls to zero and stays there, a low-trust trap. Above it, misrepresentation is rare enough that reputation grows until self-grading separates the types. Sellers who enter without an established reputation therefore have the strongest incentive to misrepresent and remain caught in the low-trust trap.

Certification enables an escape from this trap (Fig.~\ref{fig:6}B). When the item value exceeds $\bar{V}(\rho) + \underline{V}$, TP grading is available to $H$-type sellers. This matters most when the $H$-type's reputation is low. TP grading allows low-reputation $H$-type sellers to separate from low-reputation $L$-type sellers, who choose None (compare the left side of Fig.~\ref{fig:6}B with that of Fig.~\ref{fig:6}A). Reputation therefore accumulates at the full rate $\pi_0\alpha$ even when it is low, covering the range in which reputation alone cannot sustain credible self-grading. When the certification fee is small enough, reputation accumulates at every level and no trap forms, even after self-grading replaces certification as the $H$-type's preferred signal (see Methods). Certification is thus not only a static substitute for reputation in a single transaction. It can also build the reputation that later makes self-claims credible.

\begin{figure}[t!]
\centering
\includegraphics{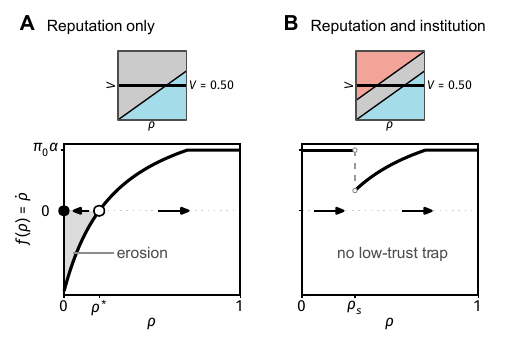}
\caption{Reputation dynamics under the two trust mechanisms. \textbf{(A)} Reputation only. \textbf{(B)} Reputation and institution. The top strip locates the slice in the $(\rho, V)$ plane (blue, Self; gray, Semi-separating; red, TP). Certification is unavailable in (A), so no TP regime arises; (B) reproduces the partition of Fig.~\ref{fig:4}. Below it is the phase line $f(\rho) = \dot{\rho}$ for the seller's reputation (see Methods). On the zero line, an open circle marks the unstable rest point, and arrows give the direction of motion. The filled circle at $\rho = 0$ marks the absorbing boundary. Shading marks the interval on which reputation erodes. In \textbf{(B)} the accumulation rate falls discontinuously at the reputation $\rho_s$ at which $H$-type sellers switch from certification to self-grading. Where the slice leaves the gray band in the top strip, the full rate is restored. Baseline parameters: $\pi_0 = 0.3$, $\phi = 0.5$, $c_{\text{TP}} = 0.2$. Both panels are drawn at $V = 0.50$.}
\label{fig:6}
\end{figure}

\section*{Discussion}

Our analyses show that reputation and third-party certification operate as complementary sources of trust, each grounded in a distinct form of costly signaling. Reputation amplifies the price premium of self-grading, a cheap-talk signal, but not that of TP grading (Fig.~\ref{fig:2}C). In the signaling game, this pattern follows from the two cost structures. Reputation collateral renders self-grading credible through the reputational cost of misrepresentation\citep{Sobel1985}. Institutional verification renders TP grading credible through an immediate certification fee that does not depend on seller reputation. The two mechanisms thus share a costly-signaling structure but operate at different times. Reputation imposes an ex-post penalty for misrepresentation, while certification imposes an ex-ante monetary cost (see ref.~\citenum{Quek2021} for this distinction in political science and refs.~\citenum{Lewis2011} and \citenum{OstromGardnerWalker1994} for other forms of ex-post cost). The two mechanisms partition the reputation--value space into three equilibrium regimes, matching the empirical regions (Fig.~\ref{fig:2}B). The mechanism that resolves the information asymmetry thus varies with seller reputation and item value. Extending the game to repeated sales shows that the two mechanisms also complement each other over time, because certification can build the reputation that later makes self-claims credible.

\subsection*{Substitution and complementarity of trust mechanisms}

While the regime-level account captures the broad pattern of signal use, the reputation $\times$ TP grading interaction admits a more specific behavioral interpretation. In the model, TP credibility is independent of seller reputation, and a high reputation makes TP grading informationally redundant. Both implications predict no such interaction. The observed interaction, however, is negative (Fig.~\ref{fig:2}C) and cannot be captured by the signaling game, so it warrants closer scrutiny. We interpret it as a strategic pricing adjustment by high-reputation sellers (see Supplementary Figs.~2 and 3 for the full analysis). Their TP listings show lower margins conditional on sale (Supplementary Fig.~3B), offset by higher sale probabilities (Supplementary Fig.~2B), leaving expected margins comparable across reputation levels (Supplementary Fig.~3D). In other words, low-reputation TP sellers sell few items at high margins, while high-reputation TP sellers sell many at low margins. This pattern echoes prior evidence that institutional verification substitutes for reputation\citep{HuiSaeediShenSundaresan2016}. Either reputation or TP grading alone can support credibility and a price premium, so the two act as substitutes in pricing terms.

Consistent with this substitution view, prior trust theory has classified trust-producing mechanisms as informal versus formal and often treated them as substitutes selected on cost\citep{Zucker1986, Kollock1999, Putnam2000}. Our analysis replaces this dichotomy with three regimes of trust production. Reputation produces trust on its own where item values are low and reputation is high, the institution produces it where item values are high and reputation is low, and, in between, reputation produces it only in part. The empirical maps show the same structure, with certification concentrated in the low-reputation, high-value corner and self-grading in the opposite corner (Figs.~\ref{fig:2}B and \ref{fig:5}). Market-wide, the two mechanisms are complements. Together they extend credible exchange across a wider range of sellers and items\citep{JinLeslie2009, Bai2025}.

Importantly, the three regimes arise for any positive reputational penalty and certification cost. Parameter changes shift only the boundaries (see Supplementary Notes~6 and 7 for comparative statics). The logic can extend beyond trading cards to any exchange that combines a reputational penalty with an institutional cost. The reputational penalty need not take the form of a feedback score. Among medieval Mediterranean traders, it took the form of collective punishment, in which no merchant would employ an agent who had cheated\citep{Greif1993, Greif1994}. The institutional cost, in turn, need not be a per-item certification fee. In the nineteenth-century United States, where migration had eroded the personal trust of settled communities, it took the form of trust purchased from certifying institutions and financial intermediaries\citep{Zucker1986}. Wherever quality is uncertain and a reputational penalty and an institutional cost exist in some form, the same three-regime structure should govern which mechanism sustains trust.

\subsection*{The scalability of trust}

The two costs in the model, the ex-ante fee and the ex-post penalty, have counterparts in biological signaling. TP grading corresponds to ex-ante cost signaling, in which honesty rests on the cost of producing the signal\citep{Zahavi1975}. Ex-post cost signaling, in turn, has been documented in primates and social insects. Signals with little or no production cost remain honest when receivers penalize deceptive signalers\citep{TibbettsDale2004, Silk2000, TibbettsIzzo2010}, so honesty rests on the ex-post cost of cheating rather than on signal production cost\citep{Hurd1995, Lachmann2001, Szamado2011}. Self-grading in our model parallels this logic, with the zero-cost signal becoming credible through the reputational penalty for misrepresentation. However, biological enforcement relies on direct dyadic interaction, confining the mechanism to small groups\citep{MaynardSmithHarper2003}. Online reputation systems overcome this scale limit by aggregating past evaluations into public scores, turning one-to-one enforcement into an information infrastructure that serves millions of anonymous participants. Thus, online reputation emerges as a social technology that renders cheap talk credible at scale\citep{Tchernichovski2019}.

The model identifies a lever and a limit for platform governance. The lever is the penalty coefficient $\phi$. The reputation-collateral bound $\bar{V}(\rho)$ rises with $\phi$, which platforms control through monitoring intensity, feedback protection, and the visibility of negative feedback\citep{Tadelis2016, BoltonGreinerOckenfels2013, OstromGardnerWalker1994}. A larger $\phi$ expands the Self regime (Fig.~\ref{fig:4}), enabling more sellers to establish credibility through reputation alone. eBay's persistent public scores and feedback protection guarantees exemplify such measures and are consistent with the wide Self-dominant region observed in our data. Even when platforms set $\phi$ high, however, reputation collateral has a structural limit. The dynamic analysis in Results suggests that, without certification, sellers who enter with little reputation have little collateral to offer and remain in the low-trust trap, a barrier to entry. TP grading can open an escape path, along which sellers move from certification to reputation-based credibility, consistent with prior evidence that certification is most valuable for sellers who lack an established reputation\citep{ElfenbeinFismanMcManus2015, DewallyEderington2006}. The model therefore suggests a developmental view of trust. Understanding both mechanisms is central to effective platform governance\citep{ChibaOkabePlotkin2026}.

\subsection*{Limitations and future directions}

Our study has several limitations. First, its cross-sectional design limits causal identification. Separating the channels behind the pricing interaction would require exogenous variation in signal availability. Second, the same design leaves the reputation dynamics untested. The low-trust trap and the escape path through certification remain model predictions, and tracing them would require longitudinal records of individual sellers. Third, we study a single collectibles market with a mature grading infrastructure. Whether the pattern holds where voluntary disclosure operates without comparable certification, as in used-car listings\citep{Lewis2011}, remains unanswered.

Taken together, our results indicate a division of labor between the two trust mechanisms in a single market. An established reputation makes a costless claim credible up to an item value that rises with that reputation. Institutional certification takes over above it. The model gives this division a general form. It reduces reputation and institutions to two types of costs, an ex-post penalty and an ex-ante fee. Platforms and societies can therefore be compared within a single frame. One natural comparison is between the United States and Japan. The two countries have been argued to differ in their relative reliance on institutions and reputation\citep{Yamagishi2011, Aoki2001}. The same pattern appears in the automotive supply chain, where trust is backed by written contracts in the United States but by trading history in Japan\citep{SakoHelper1998, Aoki1988}. This difference of degree parallels the long-standing distinction between public and private order\citep{Greif1994, Greif2006, GreifTabellini2017, GreifMokyrTabellini2025}. Thus, by bringing reputation and institutions into a common framework, this study advances our understanding of how the two mechanisms sustain trust among strangers.

\section*{Methods}

\subsection*{Data and variables}

We analyzed fixed-price (Buy It Now) listings of Pokémon trading cards on the eBay US marketplace, collected via the eBay API between June 1 and December 15, 2025. Although listings were restricted to the US marketplace on the buyer side, sellers participated internationally: 59.1\% of listings from the United States, 20.3\% from Japan, and the remaining 20.6\% from 67 countries. Each listed price was matched to an external market reference price from the PriceCharting REST API (\url{https://www.pricecharting.com/}), a publicly available price-tracking service for the trading card market, and the log price ratio was computed as the log of the listed-to-market price ratio. After removing outliers from the distribution of log price ratios using the interquartile range method (Supplementary Note~2), the final sample comprised $N = 960{,}029$ listings from 60{,}386 unique sellers. This study was approved by the ethics committee of the Graduate School of Humanities and Sociology, University of Tokyo (UTSP-25004, May 21, 2025). As the study used only publicly available marketplace data with no direct interaction with human participants, informed consent was not required.

The analyses involved three sets of variables: an outcome measure, signal categories, and seller reputation. (i) The outcome variable is the log price ratio $\log(P_i / V_i)$, where $P_i$ is the observed listing price and $V_i$ is the market reference price from PriceCharting. For TP-graded items, the grade-specific market price was used as $V_i$, so that the ratio captures pricing relative to the market price of an item with specific grading. As the market price fluctuates from day to day, we used the market price on the date of sale for sold items. For unsold listings, we used the most recent available market price at data collection as an approximation of the market price during the listing window to ensure that the ratio reflected the current market condition. For Japanese-language cards, we used a separate set of market prices in PriceCharting to account for language-specific market conditions. (ii) We distinguished between three signaling strategies: TP grading, self-grading, and none of these. Items in the TP-graded category (3.8\% of the listings) were authenticated and graded by the Professional Sports Authenticator (PSA), which assigns grades on a 10-point scale from PSA 1 (Poor) to PSA 10 (Gem Mint), at a per-card fee paid by the seller. As cards graded below PSA 7 rarely circulate in the secondary market, we restricted the TP-graded category to listings displaying PSA 7 (Near Mint) or higher as a proxy for high-grade certified items. For the cards in the self-graded category (33.3\% of the listings), sellers attached condition descriptors (e.g., ``gem mint,'' ``near mint'') to the items at no cost without obtaining PSA certification. These listings were identified through hierarchical regular-expression matching (see Supplementary Note~2 for the classification algorithm). The remaining listings were associated with no quality signal (None; 62.9\% of the listings). (iii) We constructed a seller reputation score from two eBay feedback metrics: the positive feedback percentage (the share of ratings that are positive) and the feedback score (net count of positive minus negative ratings). The percentage was logit-transformed and the score was log-transformed ($\log(\text{score} + 1)$); both were standardized before the PCA. The first principal component (78.5\% of the variance) served as the reputation score. Additional control variables included a Japanese-card indicator, a return-acceptance policy, and listing duration. Throughout the analyses, log denotes the natural logarithm.

\subsection*{Hierarchical Bayesian model}

We specified a hierarchical Bayesian model with three simultaneous equations: a multinomial logit for signal selection (TP-graded, self-graded, or None, with None as the reference category), a normal linear equation for the log price ratio, and a Bernoulli logit for purchase outcomes. Correlated seller-level random effects linked the three equations. In the signal-selection equation, item value was operationalized as the log of the ungraded reference price for all listings, including TP-graded items, on the assumption that the ungraded price serves as an exogenous factor in the signal choice. The equation also includes a reputation $\times$ item value interaction. We included additional covariates selectively, namely the Japanese-card indicator in all three equations, return-acceptance policy in the pricing and purchase equations, and log listing duration in the purchase equation only. We fitted the model in Stan using the No-U-Turn Sampler with eight chains, each run for 1,000 post-warmup iterations and thinned by a factor of two, for a total of 4,000 posterior draws (see Supplementary Notes~3 and 4 for the complete specification and Supplementary Table~1 for convergence diagnostics).

As only 13.7\% of the listings resulted in a sale, the observed transaction prices were subject to selection bias. To address this, we used the posterior draws from the Stan estimation to compute sold-conditional expected prices through posterior predictive simulation. For each signal type $k$ and reputation score, we integrated the predicted prices over the empirical covariate distribution, yielding $E[\text{price} \mid \text{reputation}, s = k, \text{sold} = 1]$. We drew 100,000 covariate vectors from the empirical distribution and combined them with 4,000 posterior draws, for a total of $4 \times 10^{8}$ simulations per signal type. We computed the AMEs of reputation as the numerical derivatives of these predictions with respect to the reputation score, averaged over the empirical distribution of reputation.

\subsection*{Signaling game}

We formulate a signaling game that formalizes how the two trust mechanisms interact. A seller privately observes product quality $\theta \in \{H, L\}$, where $P(\theta = H) = \pi_0$, and chooses a signal $s \in \{\text{TP}, \text{Self}, \text{None}\}$. Buyers observe $s$ and the seller's publicly available reputation $\rho \in [0, 1]$ (corresponding to the reputation composite defined above), and form a posterior belief $\hat{p}_s = P(H \mid s, \rho)$. We assumed that buyers were risk neutral and willing to pay up to their expected value of the item. Under the fixed-price format, the seller posts a take-it-or-leave-it price, and a buyer purchases the item whenever the posted price does not exceed the buyer's expected value. Empirically, listing prices skew toward high price premiums, while sold prices concentrate near the market price (Fig.~\ref{fig:1}B), consistent with this posted-price structure. The seller therefore posts the highest price a buyer will accept, so the equilibrium price for signal $s$ is $\hat{p}_s V$, where $V > 0$ is the full-information value of a high-quality item (a low-quality item has value zero). The signal None is treated as choosing the third option, namely listing an item without claiming any quality. It trades at the exogenous outside-option price $\pi_0 V$, independent of the seller's signal choice (see Supplementary Note~5 for the formal statement). Formally, the signaling cost $c(s, \theta, \rho)$ takes the following values across seller types and signals:

\begin{equation*}
c(s, \theta, \rho) = \begin{cases} c_{\text{TP}} & \text{if } s = \text{TP},\; \theta = H \\ +\infty & \text{if } s = \text{TP},\; \theta = L \\ 0 & \text{if } s = \text{Self},\; \theta = H \\ \phi\rho & \text{if } s = \text{Self},\; \theta = L \\ 0 & \text{if } s = \text{None} \end{cases}
\end{equation*}
where $c_{\text{TP}} > 0$ is the certification cost (ex-ante cost) and $\phi > 0$ is the penalty coefficient, which determines the penalty for dishonest self-grading, $\phi\rho$, the reputational cost. This is an ex-post cost because dishonest self-grading, once detected, triggers negative feedback, reducing the seller's reputation and diminishing future payoffs\citep{Sobel1985}. Therefore, sellers with greater reputational capital face a stronger deterrent. The key asymmetry is that $L$-type sellers cannot obtain third-party certification ($c = +\infty$), while their self-grading cost scales with reputation rather than with $V$. The solution concept is the PBE. Each seller maximizes the expected payoff given beliefs, and beliefs on the equilibrium path satisfy Bayes' rule. When multiple PBE arise, we select among them with the D1 criterion\citep{ChoKreps1987, BanksSobel1987, ChoSobel1990}.

We consider two candidate separating equilibria. In the \textit{Self-separating equilibrium}, $H$-type sellers choose Self and $L$-type sellers choose None. Bayes' rule gives the on-path belief $\hat{p}_{\text{Self}} = 1$, the price at None is fixed at $\pi_0 V$ by assumption, and the profile is a PBE if and only if $V \leq \bar{V}(\rho)$, the reputation-collateral bound (Eq.~\ref{eq:vbar}). In the \textit{TP-separating equilibrium}, $H$-type sellers choose TP and $L$-type sellers choose None. Because $L$-type sellers cannot obtain certification, Bayes' rule gives $\hat{p}_{\text{TP}} = 1$, and the profile is a PBE if and only if the TP payoff $V - c_{\text{TP}}$ is at least the None payoff $\pi_0 V$, equivalent to $V \geq \underline{V}$ (Eq.~\ref{eq:vunderline}). For $V > \bar{V}(\rho)$, the Self-separating equilibrium does not exist, and the remaining candidates are the TP-separating equilibrium and a semi-separating profile in which $H$-type sellers choose Self while $L$-type sellers mix between Self and None. For any positive reputation, no other profile, including pooling on None, is an equilibrium under D1 (Supplementary Note~7).

In the TP-separating equilibrium, Self is off the equilibrium path, Bayes' rule does not determine the belief $\hat{p}_{\text{Self}}^{\text{off}}$, and multiple PBE arise. For each type, the set of beliefs under which a deviation to Self improves on the equilibrium payoff is an interval with upper endpoint one whenever it is nonempty, because signaling costs do not depend on the buyer's posterior and both deviation payoffs increase in the belief with the common slope $V$. The intervals begin at $\bar{p}_H = 1 - c_{\text{TP}}/V$ for the $H$-type and $\bar{p}_L = \pi_0 + \phi\rho/V$ for the $L$-type, so the D1 comparison reduces to the ordering of the two thresholds (Supplementary Note~7). When $\bar{p}_L > \bar{p}_H$, the criterion sets $\hat{p}_{\text{Self}}^{\text{off}} = 1$. When $\bar{p}_H > \bar{p}_L$, it sets $\hat{p}_{\text{Self}}^{\text{off}} = 0$. The thresholds are equal at $V = (c_{\text{TP}} + \phi\rho)/(1 - \pi_0) = \bar{V}(\rho) + \underline{V}$, the certification boundary. Because $\bar{p}_H$ increases in $V$ while $\bar{p}_L$ decreases, $\bar{p}_L > \bar{p}_H$ holds exactly for $V$ below this boundary.

Under the assigned beliefs, the TP-separating equilibrium survives only above the certification boundary. Below it, the belief $\hat{p}_{\text{Self}}^{\text{off}} = 1$ makes the deviation to Self profitable for the $H$-type. In the semi-separating profile, the $L$-type's indifference between Self and None, $\hat{p}_{\text{Self}} V - \phi\rho = \pi_0 V$, pins the on-path belief at $\hat{p}_{\text{Self}} = \pi_0 + \phi\rho/V$, which lies strictly below one for $V > \bar{V}(\rho)$, so a self-claim is only partially credible. This profile exists exactly for $\bar{V}(\rho) \leq V \leq \bar{V}(\rho) + \underline{V}$. The equilibrium is therefore Self-separating for $V < \bar{V}(\rho)$, semi-separating between the two boundaries, and TP-separating above the certification boundary, and for any positive reputation it is unique off the boundaries (Supplementary Note~7). The two parallel boundaries partition the $(\rho, V)$ space into the Self, Semi-separating, and TP regimes described in Results (Fig.~\ref{fig:4}). The baseline parameter values are $\pi_0 = 0.3$, $\phi = 0.5$, and $c_{\text{TP}} = 0.2$ (Table~\ref{tab:1} and Supplementary Note~9; see Supplementary Note~6 for comparative statics).

\subsection*{Reputation dynamics}

Consider a seller for whom certification is unavailable, so that Self and None are the only signals. The certification boundary comes from the $H$-type deviation to TP, so for this seller the semi-separating equilibrium holds at every $V > \bar{V}(\rho)$. In that equilibrium, $H$-type sellers choose Self, and $L$-type sellers choose Self with the probability $\sigma_L(\rho)$ that keeps them indifferent between Self and None. Substituting the on-path belief $\hat{p}_{\text{Self}} = \pi_0 + \phi\rho/V$ into Bayes' rule gives
\begin{equation*}
\sigma_L(\rho) = \frac{\pi_0\,(1 - \pi_0 - \phi\rho/V)}{(1 - \pi_0)(\pi_0 + \phi\rho/V)} .
\end{equation*}

Fixing the item value $V$ leaves a single equation for $\rho$. Trade on the $H$-type share $\pi_0$ of the seller's listings raises reputation. On the remaining share, misrepresentation occurs with probability $\sigma_L(\rho)$ and lowers it by the same amount, while listings that carry no claim leave it unchanged, so a single rate $\alpha > 0$ governs both directions. The law of motion is therefore
\begin{equation*}
\dot{\rho} = \alpha\,[\,\pi_0 - (1 - \pi_0)\,\sigma_L(\rho)\,] .
\end{equation*}

The law of motion has at most one interior rest point, $\rho^{*} = V(1/2 - \pi_0)/\phi$, and that rest point is unstable (see Supplementary Note~8 for the derivation and the stability analysis). It is interior when high-quality listings are the minority, $\pi_0 < 1/2$, and when $V(1/2 - \pi_0) < \phi$. At the baseline values the second condition is slack, so the minority condition alone decides whether the rest point exists. Absent certification, reputation starting below $\rho^{*}$ is absorbed at $\rho = 0$.

With certification available, the $H$-type chooses TP below the switching reputation $\rho_s = [(1 - \pi_0)V - c_{\text{TP}}]/\phi$, which is interior exactly for $\underline{V} < V < (\phi + c_{\text{TP}})/(1 - \pi_0)$, and below it reputation grows at the full rate $\pi_0\alpha$. At $\rho_s$ misrepresentation resumes and the accumulation rate falls. Whether that drop stalls reputation turns only on the sign of $V - 2c_{\text{TP}}$ (Supplementary Note~8). When $V > 2c_{\text{TP}}$ reputation accumulates at every level and no trap forms. When $V < 2c_{\text{TP}}$ the rest point lies above the switch. Reputation below $\rho^{*}$ then converges on $\rho_s$ from either side, so the trap is the switch itself rather than a rest point of the law of motion. Reputation above $\rho^{*}$ recovers the full rate.

\section*{Data availability}

The processed data needed to reproduce the figures are available at the Open Science Framework\cite{KidoOhtsubo2026_OSF} (\url{https://doi.org/10.17605/OSF.IO/6HJFK}). Raw listing data were obtained from the eBay API and raw market-price data from the PriceCharting API. A full description of the data-collection procedure is provided in Supplementary Note~1, to enable independent re-collection.

\section*{Code availability}

All analysis code needed to reproduce the figures is available at the Open Science Framework\cite{KidoOhtsubo2026_OSF} (\url{https://doi.org/10.17605/OSF.IO/6HJFK}).

\bibliographystyle{unsrtnat}
\bibliography{references}

\section*{Acknowledgements}

We thank Hirokazu Hatta for valuable discussions and advice throughout this project. We thank Kiyoshi Izumi and the members of his laboratory for their helpful comments. We also thank the participants in our laboratory seminar for their feedback. We thank Scribendi (\url{https://www.scribendi.com/}) for English-language editing. This work was supported by a Grant-in-Aid for JSPS Fellows (JSPS KAKENHI Grant Number JP25KJ0080 to Y.K.) and a Grant-in-Aid for Transformative Research Areas (A) (MEXT KAKENHI Grant Number 24H02200 to Y.O.).

\section*{Author contributions}

Y.K. and Y.O. designed research; Y.K. performed research and analyzed data; and Y.K. and Y.O. wrote the paper.

\section*{Competing interests}

The authors declare no competing interests.

\section*{Additional information}

\noindent\textbf{Supplementary information} is available for this paper.

\noindent\textbf{Correspondence and requests for materials} should be addressed to Yuta Kido.

\clearpage
\appendix
\setcounter{figure}{0}\renewcommand{\thefigure}{S\arabic{figure}}
\setcounter{table}{0}\renewcommand{\thetable}{S\arabic{table}}
\setcounter{equation}{0}\renewcommand{\theequation}{S\arabic{equation}}
\setcounter{section}{0}
\section*{Supplementary Materials}
\bigskip
\noindent\textbf{Organization of this Supplementary Information:}
\begin{enumerate}[nosep,leftmargin=2em]
\item \textbf{Supplementary Notes} (pp.~\pageref{sec:si-methods}--\pageref{sec:si-methods-end}). Supplementary Notes 1 to 4 describe the data pipeline, the variables, the full specification of the hierarchical Bayesian model, and its estimation. Supplementary Notes 5 to 9 state the assumptions of the signaling game, derive the Self-separating, semi-separating, and third-party (TP) separating equilibria with the D1 refinement, analyze the reputation dynamics, and give the parameter values used in the figures.
\item \textbf{Supplementary Figures and Tables} (pp.~\pageref{sec:si-results-top}--\pageref{sec:si-results-end}). Descriptive price analysis (Supplementary Fig.~1), seller volume (Supplementary Fig.~2), revenue and margin (Supplementary Fig.~3), the continuous regime boundary in reputation--value space (Supplementary Fig.~4), convergence diagnostics (Supplementary Table~1), model fit (Supplementary Fig.~5), full coefficient estimates (Supplementary Table~2), and signal composition (Supplementary Fig.~6).
\end{enumerate}

\clearpage

\SItext

\refstepcounter{section}\section*{Empirical Methods}\label{sec:si-empirical}

\subsection{Data Pipeline}\label{subsec:emp-A}

Listing data were collected daily from the eBay US marketplace through the Browse API (version 1), and market reference prices were collected daily from the PriceCharting REST API, which reports historical card prices by grade. Collection spanned 198 days (June 1 to December 15, 2025) and was restricted to the Pokémon trading card category and to listings offered to buyers on the US marketplace, which holds shipping costs and transaction conditions constant across listings. Sellers were located in 69 countries. The United States accounted for 59.1\% of listings, Japan for 20.3\%, Canada for 6.7\%, the United Kingdom for 6.2\%, Australia for 4.4\%, France for 1.6\%, and 63 further countries for a combined 1.7\%.

Listings were retrieved with the single query ``pokemon card'' in eBay's category for individual trading cards, restricted to fixed-price listings. New listings were collected every six hours, and listings that accept returns were collected in a separate run. Market prices were downloaded once a day as a full file of about 74,000 card--grade entries, and Japanese-language cards were identified by the ``Japanese'' tag in the product name.

Each listing was matched to a PriceCharting entry in four steps. Listings for bulk lots were excluded, the card number and set size were extracted from the title by regular expression, the set was identified from the title, and the card number was verified before the listing was joined to the entry for that set, number, and grade.

\subsection{Variable Construction}\label{subsec:emp-B}

\subsubsection{Price premium}

The price premium is the log ratio of the listing price to the market reference price,

\begin{equation}
\widetilde{p}_i = \log(P_i / V_i) \label{eq:s1}
\end{equation}

where $P_i$ is the listing price in USD and $V_i$ is the PriceCharting reference price for the corresponding card--grade combination. For sold items, $V_i$ is the market price on the date of sale; for unsold listings, it is the most recent market price.

Outliers were removed with Tukey's interquartile-range rule applied to the log price ratio, excluding observations below $Q_1 - 1.5 \times \mathrm{IQR}$ or above $Q_3 + 1.5 \times \mathrm{IQR}$. This excluded 23,421 observations (2.3\%) and reduced the sample from 1,018,321 to 994,900 listings. Restriction to complete cases gave the analysis sample of $N = 960{,}029$.

\subsubsection{Seller reputation score}

eBay reports two seller reputation metrics, the Feedback Score (positive minus negative ratings, cumulative over the seller's history) and the Positive Feedback Percentage (the share of positive ratings over the preceding 12 months, from 0 to 100\%). We aggregated both metrics to seller-level means, applied a logit transformation to the percentage and a log transformation to the score to address ceiling effects and skewness, standardized both, and took the first principal component as the composite reputation score.

The first component explained 78.5\% of the variance, and the score had $M = 0.013$ and $SD = 1.188$ across the 60,386 sellers in the analysis sample. The two inputs differed in distribution. The Positive Feedback Percentage averaged 86.3\% ($SD = 33.6\%$). The Feedback Score ranged from $-3$ to 672,427 and was bimodal on the log scale; the 23 sellers (0.04\%) with negative scores were set to zero before the log transformation.

\subsubsection{Signal categories}

Each listing carries one of three signals: None (no quality signal; 62.9\% of listings), Self-graded (a self-claimed condition description; 33.3\%), and TP-graded (third-party certification; 3.8\%).

TP-graded listings are cards certified by Professional Sports Authenticator (PSA), the largest TP grading service for Pokémon trading cards by transaction volume, at grades 10 (Gem Mint), 9 (Mint), 8 (Near Mint--Mint), and 7 (Near Mint); the criteria for centering, surface, and corners become stricter at each step. A listing was classified as TP-graded when eBay's condition field marked it as graded and its title named a PSA grade of 7 to 10.

Self-graded listings are ungraded cards whose titles contain condition descriptors equivalent to these PSA grades. Descriptors were matched in a fixed order, and the first match was taken: PSA 10 (``gem mint'', ``perfect'', or ``flawless''), PSA 8 (``near mint'' or ``NM'' followed by ``mint'' or ``MT''), PSA 7 (``near mint'' or ``NM'' not followed by ``mint'' or ``MT''), and PSA 9 (``mint'' or ``MT'' not followed by ``gem''). Titles without a match were classified as None.

\subsubsection{Control variables}

Three control variables enter the model: an indicator for Japanese-language cards (36.5\% of listings), an indicator for whether returns are accepted (48.8\%), and the number of days since the listing was created ($M = 258.6$, $SD = 355.2$). They control for card-level heterogeneity, seller listing policy, and the duration of unsold listings.

\subsection{Full Model Specification}\label{subsec:emp-C}

The hierarchical Bayesian model (HBM) fitted to the analysis sample consists of three simultaneously estimated equations, a multinomial logit for signal choice (None, Self-graded, TP-graded), a normal linear equation for the log price ratio, and a Bernoulli logit for the sale outcome, linked by correlated seller-level random effects. Joint estimation with correlated random effects accounts for the seller's endogenous signal choice and for selection in observed prices.

The signal choice follows a multinomial logit,

\begin{equation}
P(\text{cert}_i = k) = \frac{\exp(\eta_i^{(k)})}{\sum_{k'} \exp(\eta_i^{(k')})} \label{eq:s2}
\end{equation}

with $\eta_i^{(\text{None})} = 0$ as the reference category. For $k \in \{\text{Self-graded}, \text{TP-graded}\}$, the linear predictor is

\begin{equation}
\eta_i^{(k)} = \alpha^{(k)} + \gamma_\rho^{(k)} \cdot \text{reputation}_j + \delta^{(k)} \cdot \log(V_i^{\text{ung}}) + \psi^{(k)} \cdot \text{reputation}_j \cdot \log(V_i^{\text{ung}}) + \zeta^{(k)} \cdot \text{JP}_i + \lambda^{(k)} \cdot u_j^{signal} \label{eq:s3}
\end{equation}

where $\text{reputation}_j$ is the seller reputation score, $\log V_i^{\text{ung}}$ is the log ungraded market price (distinct from the grade-specific $V_i$ of Supplementary Note~2), $\text{JP}_i$ is the Japanese-card indicator, and $u_j^{signal}$ is a seller random effect with factor loadings $\lambda^{\text{Self}} = 1$ (fixed) and $\lambda^{\text{TP}}$ (estimated).

The log price ratio is normally distributed,

\begin{equation}
\log(P_i / V_i) \sim \mathcal{N}(\mu_i^{price}, \sigma_{price}^2) \label{eq:s4}
\end{equation}

with conditional mean

\begin{equation}
\mu_i^{price} = \alpha^{price} + \mathbf{X}_i^{price} \boldsymbol{\beta}^{price} + \beta_{signal}^{\text{Self}} \cdot \mathbf{1}_{\text{Self}} + \beta_{signal}^{\text{TP}} \cdot \mathbf{1}_{\text{TP}} + u_j^{price} \label{eq:s5}
\end{equation}

The design matrix $\mathbf{X}^{price}$ contains five covariates: the seller reputation score, the Japanese-card indicator, reputation $\times$ self-graded, reputation $\times$ TP-graded, and the returns-accepted indicator. The interaction coefficients $\beta_3$ and $\beta_4$ are the parameters of primary interest; their posterior estimates are reported in Supplementary Table~\ref{tab:coef}.

The sale outcome is Bernoulli with a logit link,

\begin{equation}
\text{sold}_i \sim \text{Bernoulli}(\text{logit}^{-1}(\eta_i^{sold})) \label{eq:s6}
\end{equation}

with linear predictor

\begin{equation}
\eta_i^{sold} = \alpha^{sold} + \mathbf{X}_i^{sold} \boldsymbol{\beta}^{sold} + u_j^{sold} \label{eq:s7}
\end{equation}

The design matrix $\mathbf{X}^{sold}$ contains nine covariates: the seller reputation score, the self-graded indicator, the TP-graded indicator, the Japanese-card indicator, the returns-accepted indicator, the log price ratio, the log listing duration, reputation $\times$ self-graded, and reputation $\times$ TP-graded.

The three seller random effects are drawn jointly in a non-centered parameterization, which improves sampling efficiency for the 60,386 sellers in the analysis sample,

\begin{equation}
\mathbf{u}_j = \begin{pmatrix} u_j^{price} \\ u_j^{sold} \\ u_j^{signal} \end{pmatrix} = L \cdot \mathbf{z}_j, \quad \mathbf{z}_j \sim \mathcal{N}(\mathbf{0}, I_3) \label{eq:s8}
\end{equation}

where $L = \text{diag}(\boldsymbol{\sigma}_u) \cdot L_\Omega$ is the Cholesky factor of the covariance matrix, $\boldsymbol{\sigma}_u$ contains the standard deviations, and $L_\Omega$ is the Cholesky factor of the correlation matrix.

All parameters receive weakly informative priors. In the signal equation, the intercepts have $\alpha^{(k)} \sim \mathcal{N}(0, 2)$, the slopes $\gamma_\rho^{(k)}, \delta^{(k)}, \zeta^{(k)} \sim \mathcal{N}(0, 1)$, the interactions $\psi^{(k)} \sim \mathcal{N}(0, 0.5)$, and the TP factor loading $\lambda^{\text{TP}} \sim \mathcal{N}(1, 1)$. In the price equation, $\alpha^{price} \sim \mathcal{N}(0, 2)$, $\boldsymbol{\beta}^{price} \sim \mathcal{N}(0, 2)$ ($K = 5$), $\beta_{signal}^{\text{Self}}, \beta_{signal}^{\text{TP}} \sim \mathcal{N}(0, 1)$, and $\sigma_{price} \sim \text{Exponential}(1)$. In the sale equation, $\alpha^{sold} \sim \mathcal{N}(0, 2)$ and $\boldsymbol{\beta}^{sold} \sim \mathcal{N}(0, 2)$ ($K = 9$). For the random effects, $\mathbf{z}_j \sim \mathcal{N}(\mathbf{0}, I_3)$, $\sigma_u \sim \text{Exponential}(1)$ for each of the three elements, and $L_\Omega \sim \text{LKJ}(2)$.

\subsection{Estimation}\label{subsec:emp-D}

We estimated the model by Markov chain Monte Carlo (MCMC) in Stan, using the No-U-Turn Sampler (NUTS), a variant of Hamiltonian Monte Carlo. Eight chains were run for 1,000 warm-up and 1,000 sampling iterations each, with thinning by a factor of 2, a step-size adaptation target of 0.95, and a maximum tree depth of 12, giving 4,000 posterior draws. Convergence diagnostics for all reported quantities are given in Supplementary Table~\ref{tab:convergence} and predictive fit in Supplementary Fig.~\ref{fig:S5}.

\FloatBarrier
\clearpage
\refstepcounter{section}\section*{Theoretical Model}\label{sec:si-theory}

\subsection{Model Assumptions}\label{subsec:th-A}

The model is a signaling game between a seller (Sender) and buyers (Receivers). The seller privately observes product quality and chooses a signal $s \in \{\text{TP}, \text{Self}, \text{None}\}$, where TP denotes third-party grading; buyers then form beliefs and decide whether to buy at the posted price. Six assumptions define the environment, Supplementary Notes~6 and 7 derive the equilibria from them, and Supplementary Note~8 derives the reputation dynamics.

\begin{assumption}[Quality and item value]
Product quality is binary, $\theta \in \{H, L\}$, with item values $v(H) = V > 0$ and $v(L) = 0$. The prior probability of high quality is $P(\theta = H) = \pi_0 \in (0, 1)$.
\end{assumption}

\begin{assumption}[Information structure]
The seller's type $\theta$ is private information. The seller's reputation $\rho$ is public. Buyers observe the pair $(s, \rho)$ and form beliefs $\hat{p}_s$.
\end{assumption}

\begin{assumption}[Cost structure]
Signaling costs depend on type and reputation,
\begin{equation}
c(s, \theta, \rho) = \begin{cases} c_{\text{TP}} & s = \text{TP},\ \theta = H \\ +\infty & s = \text{TP},\ \theta = L \\ 0 & s = \text{Self},\ \theta = H \\ \phi\rho & s = \text{Self},\ \theta = L \\ 0 & s = \text{None} \end{cases} \label{eq:s10}
\end{equation}
where $c_{\text{TP}} > 0$ is the TP grading fee, $\phi > 0$ is the penalty coefficient for detected misrepresentation, and $\rho \in [0,1]$ is the seller's reputation, based on cumulative transaction history.

Equation~\ref{eq:s10} embodies four further assumptions: $L$-types cannot pass TP grading, truthful self-grading is costless, a false self-claim triggers a penalty proportional to reputation, and detection is certain ($q = 1$). Under imperfect detection with $0 < q < 1$, the reputation-collateral bound $\bar{V}(\rho)$ of Supplementary Note~6 scales by $q$, and the regime partition and comparative statics are otherwise unchanged.
\end{assumption}

\begin{assumption}[Outside option for None]
The signal None denotes non-participation in the signaling game. It trades at the exogenous outside-option price $\pi_0 V$ and is not subject to Bayesian updating.
\end{assumption}

\begin{assumption}[Posted prices]
Buyers are risk-neutral and willing to pay up to their expected value of the item. Given signal $s$ and reputation $\rho$, buyers hold the belief $\hat{p}_s = P(H \mid s, \rho)$, and the seller posts the highest price a buyer accepts, so the equilibrium price is $p(s, \rho) = \hat{p}_s V$ for $s \in \{\text{TP}, \text{Self}\}$. The price at None is exogenous (Assumption 4).
\end{assumption}

\begin{assumption}[Seller payoff]
The seller is risk-neutral, and the payoff is the price received minus the signaling cost,
\begin{equation}
\Pi_{\text{seller}}(\theta, s, \rho) = p(s, \rho) - c(s, \theta, \rho) \label{eq:s9}
\end{equation}
\end{assumption}

\subsubsection*{Payoffs}

Assumptions 1 to 6 give the following payoffs: $H \to \text{TP}$ yields $\hat{p}_{\text{TP}} V - c_{\text{TP}}$; $L \to \text{TP}$ is infeasible; $H \to \text{Self}$ yields $\hat{p}_{\text{Self}} V$; $L \to \text{Self}$ yields $\hat{p}_{\text{Self}} V - \phi\rho$; and either type at None receives $\pi_0 V$. The belief $\hat{p}_{\text{TP}} = 1$ follows from Assumption 3, whereas $\hat{p}_{\text{Self}}$ is determined in equilibrium (Supplementary Note~6).

\subsubsection*{Solution concept}

\begin{definition}[Strategy]
A seller's strategy maps type and reputation to a probability distribution over signals,
\[\sigma: \Theta \times [0, 1] \to \Delta(S),\]
where $\Theta = \{H, L\}$, $S = \{\text{TP}, \text{Self}, \text{None}\}$, and $\Delta(S)$ is the set of probability distributions over $S$.
\end{definition}

\begin{definition}[Belief system]
The buyer's belief system assigns to each signal--reputation pair a probability that the item is high quality,
\[\hat{p}: S \times [0, 1] \to [0, 1].\]
At TP and Self, $\hat{p}(s, \rho) = P(H \mid s, \rho)$ is a posterior; at None, $\hat{p}$ is fixed at the prior $\pi_0$ (Assumption 4).
\end{definition}

\begin{definition}[Equilibrium]
A strategy--belief pair $(\sigma^*, \hat{p}^*)$ is an equilibrium if two conditions hold.

(i) \textit{Sequential rationality}: for each $(\theta, \rho)$, $\sigma^*(\theta, \rho)$ places positive probability only on signals in
\[\arg\max_{s \in S} \left[ p(s, \rho) - c(s, \theta, \rho) \right],\]
where $p(s, \rho)$ is the price that Assumptions 4 and 5 assign to the belief $\hat{p}^*(s, \rho)$.

(ii) \textit{Bayes consistency}: the beliefs at Self and TP satisfy Bayes' rule whenever those signals are on the equilibrium path. None is exempt (Assumption 4).
\end{definition}

\begin{definition}[Information sets]
$I_{\text{TP}}$ is a singleton with $\hat{p}_{\text{TP}} = 1$, because $L$-type sellers cannot obtain TP grading (Assumption 3). $I_{\text{Self}}$ may contain both types, so $\hat{p}_{\text{Self}} \in [0, 1]$ is determined in equilibrium. $I_{\text{None}}$ carries the prior, $\hat{p}_{\text{None}} = \pi_0$ (Assumption 4).
\end{definition}

\subsection{Self-Separating Equilibrium}\label{subsec:th-B}

The payoffs in Supplementary Note~5 leave $\hat{p}_{\text{Self}}$ undetermined. Here we show that the profile in which $H$-type sellers choose Self and $L$-type sellers choose None,
\[\sigma^*(H, \rho) = \text{Self}, \quad \sigma^*(L, \rho) = \text{None},\]
is a perfect Bayesian equilibrium (PBE) whenever $V \leq \bar{V}(\rho)$, and we derive $\bar{V}(\rho)$ from the $L$-type's incentive constraint. The profile is suggested by the cost asymmetry in Assumption 3: Self is costless for the $H$-type and costs $\phi\rho$ for the $L$-type.

\subsubsection*{Beliefs and payoffs}

Because only $H$-type sellers choose Self on the equilibrium path, Bayes' rule gives $\hat{p}_{\text{Self}} = P(H \mid \text{Self}) = 1$. The belief at TP is $\hat{p}_{\text{TP}} = 1$ (Assumption 3), and the price at None is $\pi_0 V$ (Assumption 4). On the path, $H \to \text{Self}$ yields $V$ and $L \to \text{None}$ yields $\pi_0 V$. The deviations yield $V - \phi\rho$ for $L \to \text{Self}$, $\pi_0 V$ for $H \to \text{None}$, and $V - c_{\text{TP}}$ for $H \to \text{TP}$.

\subsubsection*{Incentive compatibility}

The $H$-type prefers Self to None because $V > \pi_0 V$ for every $\pi_0 \in (0, 1)$, and prefers Self to TP because $c_{\text{TP}} > 0$. Both conditions hold under Assumptions 1 and 3. The binding constraint is the $L$-type's. None is preferred to Self when $\pi_0 V > V - \phi\rho$, that is, when $V(1 - \pi_0) < \phi\rho$. Solving for $V$ defines the reputation-collateral bound,
\begin{equation}
V < \frac{\phi}{1 - \pi_0}\,\rho \equiv \bar{V}(\rho) \label{eq:bar-V}
\end{equation}
The constraint holds strictly for $V < \bar{V}(\rho)$ and fails for $V > \bar{V}(\rho)$, where the $L$-type deviates to Self. At $V = \bar{V}(\rho)$ the $L$-type is indifferent and None remains a best response. The equilibrium therefore exists exactly for $V \leq \bar{V}(\rho)$.

\begin{definition}[Reputation-collateral bound]
$\bar{V}(\rho) = \frac{\phi}{1-\pi_0}\,\rho$ is the highest item value at which a seller with reputation $\rho$ sustains credible self-grading through reputation alone. The slope $\phi/(1-\pi_0)$ is the additional item value sustained per unit of reputation. This bound is Eq.~1 in the main text.
\end{definition}

\begin{proposition}[Existence of the Self-separating equilibrium]\label{prop:self-existence}
The separating equilibrium $(H \to \text{Self},\ L \to \text{None})$ exists if and only if $V \leq \bar{V}(\rho)$. At $V = \bar{V}(\rho)$ the $L$-type is indifferent between None and Self, and the profile coincides with the semi-separating equilibrium of Supplementary Note~7 at $\sigma_L = 0$.
\end{proposition}

\subsubsection*{Comparative statics}

The reputation-collateral bound responds to the primitives as follows:
\begin{align}
\frac{\partial \bar{V}(\rho)}{\partial \rho}   &= \frac{\phi}{1 - \pi_0} > 0, \label{eq:dvbar-rho}\\
\frac{\partial \bar{V}(\rho)}{\partial \phi}   &= \frac{\rho}{1 - \pi_0} \geq 0, \label{eq:dvbar-phi}\\
\frac{\partial \bar{V}(\rho)}{\partial \pi_0}  &= \frac{\phi\rho}{(1 - \pi_0)^2} \geq 0. \label{eq:dvbar-pi0}
\end{align}
The first derivative is positive at every reputation, and the other two are positive for $\rho > 0$. Higher reputation (Eq.~\ref{eq:dvbar-rho}), a harsher penalty $\phi$ (Eq.~\ref{eq:dvbar-phi}), and a higher prior $\pi_0$ (Eq.~\ref{eq:dvbar-pi0}) each expand the Self regime, because more reputation collateral, stronger sanctions, and a smaller deception rent $(1-\pi_0)V$ all raise the bound.

\subsection{TP-Separating Equilibrium and Equilibrium Refinement}\label{subsec:th-C}

For $V \leq \bar{V}(\rho)$, Bayes' rule pins down the belief at Self (Supplementary Note~6) and Assumption 4 fixes the price at None. For $V > \bar{V}(\rho)$ the Self-separating equilibrium does not exist, and the candidate that uses TP grading as the quality signal makes Self an off-path signal whose belief Bayes' rule leaves undetermined. This note derives the TP-separating and semi-separating equilibria and selects among candidates with the D1 criterion\citep{ChoKreps1987, BanksSobel1987, ChoSobel1990}, which restricts off-path beliefs where the Intuitive Criterion does not.

\subsubsection*{Equilibrium candidate}

\[\sigma^*(H, \rho) = \text{TP}, \quad \sigma^*(L, \rho) = \text{None}\]

On the equilibrium path, $\hat{p}_{\text{TP}} = 1$ because $L$-type sellers cannot obtain TP grading (Assumption 3), and the price at None is $\pi_0 V$ (Assumption 4). Self is off the path, so under PBE the belief $\hat{p}_{\text{Self}}^{\text{off}} \in [0, 1]$ is unrestricted, and we resolve this indeterminacy below.

\subsubsection*{Deviation thresholds}

Whether a type gains from a deviation to Self depends on $\hat{p}_{\text{Self}}^{\text{off}}$. The $H$-type receives $V - c_{\text{TP}}$ under TP and $\hat{p}_{\text{Self}}^{\text{off}} V$ from Self, so it gains when
\[\hat{p}_{\text{Self}}^{\text{off}} > 1 - \frac{c_{\text{TP}}}{V}.\]
The $L$-type receives $\pi_0 V$ under None and $\hat{p}_{\text{Self}}^{\text{off}} V - \phi\rho$ from Self, so it gains when
\[\hat{p}_{\text{Self}}^{\text{off}} > \pi_0 + \frac{\phi\rho}{V}.\]
Write $\bar{p}_H = 1 - c_{\text{TP}}/V$ and $\bar{p}_L = \pi_0 + \phi\rho/V$ for the two thresholds. Because $c_{\text{TP}} > 0$, $\bar{p}_H < 1$, so sufficiently optimistic beliefs always make the $H$-type's deviation profitable.

\begin{lemma}[Reduction to threshold comparison]\label{lem:threshold}
Let Self be off the equilibrium path and $V \geq \underline{V}$, the range on which the TP-separating profile is a PBE (below). The price is strictly increasing in the belief (Assumption 5) and signaling costs do not depend on it (Assumption 3), so type $\theta$'s deviation payoff $\hat{p}_{\text{Self}}^{\text{off}} V - c(\text{Self}, \theta, \rho)$ is affine in $\hat{p}_{\text{Self}}^{\text{off}}$ with slope $V > 0$. The set of beliefs under which type $\theta$ strictly gains is $D_\theta = (\bar{p}_\theta, 1] \cap [0, 1]$, and the set under which it is indifferent is $D^0_\theta = \{\bar{p}_\theta\} \cap [0, 1]$. On this range $\bar{p}_H \geq \pi_0 > 0$ and $\bar{p}_L > \pi_0 > 0$, so $D_\theta \cup D^0_\theta = [\bar{p}_\theta, 1]$ whenever $\bar{p}_\theta \leq 1$ and is empty otherwise, and $D_H$ is nonempty because $\bar{p}_H < 1$. Consequently, $D_\theta \cup D^0_\theta \subsetneq D_{\theta'}$ holds if and only if $\bar{p}_\theta > \bar{p}_{\theta'}$.
\end{lemma}

\subsubsection*{PBE existence}

The $H$-type prefers TP to None when $V - c_{\text{TP}} > \pi_0 V$, that is, when
\begin{equation}
V > \frac{c_{\text{TP}}}{1 - \pi_0} \equiv \underline{V} \label{eq:underline-V}
\end{equation}
with indifference at $V = \underline{V}$. This bound involves no off-path belief and does not depend on $\rho$. It is the lowest item value at which the quality premium $(1 - \pi_0)V$ covers the fee $c_{\text{TP}}$, and it is Eq.~2 in the main text. The credibility of TP grading rests on the grading institution rather than on the seller's reputation. No type gains from a deviation to Self when $\hat{p}_{\text{Self}}^{\text{off}} \leq \min(\bar{p}_H, \bar{p}_L)$, so the TP-separating profile is a PBE for every $V \geq \underline{V}$, supported by any off-path belief in this range.

\subsubsection*{Refinement of the off-path belief}

Self is the only signal to which the refinement applies, because the belief at None is fixed at the prior (Assumption 4) and the belief at TP equals one (Assumption 3). The D1 criterion places zero probability on a type whose set of deviation-supporting beliefs is strictly contained in that of the other type. By Lemma~\ref{lem:threshold}, this comparison reduces to the ordering of $\bar{p}_H$ and $\bar{p}_L$.

\begin{definition}[Refined off-path belief]\label{def:d1-belief}
When $\bar{p}_L > \bar{p}_H$, every belief under which the $L$-type gains from a deviation to Self also lets the $H$-type gain strictly (Lemma~\ref{lem:threshold}), and the off-path belief is set to $\hat{p}_{\text{Self}}^{\text{off}} = 1$. When $\bar{p}_H > \bar{p}_L$, the reverse containment holds and the belief is set to $\hat{p}_{\text{Self}}^{\text{off}} = 0$. When $\bar{p}_H = \bar{p}_L$, the two sets coincide and the belief is not restricted.
\end{definition}

The two thresholds are equal when $1 - c_{\text{TP}}/V = \pi_0 + \phi\rho/V$, that is, when $V = (c_{\text{TP}} + \phi\rho)/(1 - \pi_0) = \bar{V}(\rho) + \underline{V}$, the certification boundary. Because $\bar{p}_H$ increases in $V$ while, for $\rho > 0$, $\bar{p}_L$ decreases in $V$, three cases cover the parameter space away from this value.

\textit{(a) $V \leq \bar{V}(\rho)$.} Here $\bar{p}_L \geq 1$, so $D_L \cup D^0_L$ is empty or equals $\{1\}$, while $D_H$ is nonempty. The belief is set to $\hat{p}_{\text{Self}}^{\text{off}} = 1$, the $H$-type receives $V$ from Self against $V - c_{\text{TP}}$ from TP, and the TP-separating profile is not an equilibrium.

\textit{(b) $\bar{V}(\rho) < V < \bar{V}(\rho) + \underline{V}$.} Here $\bar{p}_H < \bar{p}_L < 1$, so $D_L \cup D^0_L = [\bar{p}_L, 1] \subsetneq (\bar{p}_H, 1] = D_H$. The belief is set to $\hat{p}_{\text{Self}}^{\text{off}} = 1$, the $H$-type receives $V > V - c_{\text{TP}}$, and the TP-separating profile is not an equilibrium.

\textit{(c) $V > \bar{V}(\rho) + \underline{V}$.} Here $\bar{p}_L < \bar{p}_H$, so $D_H \cup D^0_H = [\bar{p}_H, 1] \subsetneq (\bar{p}_L, 1] = D_L$. The belief is set to $\hat{p}_{\text{Self}}^{\text{off}} = 0$. The $H$-type receives $0$ from Self against $V - c_{\text{TP}} > 0$ from TP, the $L$-type receives $-\phi\rho$ from Self against $\pi_0 V$ from None, and the TP-separating profile is an equilibrium.

The reversal between (b) and (c) follows from the two cost structures. The $H$-type's gain from an unexpected self-claim is the fee $c_{\text{TP}}$, which does not vary with $V$. The $L$-type's gain from misrepresentation is proportional to $V$, while the collateral $\phi\rho$ does not vary with $V$. Below $\bar{V}(\rho) + \underline{V}$ the deviation threshold is lower for the $H$-type; above it, lower for the $L$-type.

\begin{proposition}[Existence of the TP-separating equilibrium]\label{prop:tp-existence}
The profile $(H \to \text{TP},\ L \to \text{None})$:
(i) is a PBE if and only if $V \geq \underline{V}$, supported by any off-path belief $\hat{p}_{\text{Self}}^{\text{off}} \leq \min(\bar{p}_H, \bar{p}_L)$, with the $H$-type indifferent between TP and None at $V = \underline{V}$;
(ii) for $V \neq \bar{V}(\rho)$, satisfies the Intuitive Criterion if and only if, in addition, $V > \bar{V}(\rho)$;
(iii) for $V \neq \bar{V}(\rho) + \underline{V}$, satisfies the D1 criterion if and only if $V > \bar{V}(\rho) + \underline{V}$.
Since $\bar{V}(\rho) + \underline{V} > \max(\bar{V}(\rho), \underline{V})$ for every $\rho > 0$, (iii) implies (i) and (ii).
\end{proposition}

Part (ii) holds because the $L$-type's maximal deviation payoff $V - \phi\rho$ falls below its equilibrium payoff $\pi_0 V$ exactly when $V < \bar{V}(\rho)$, in which case the Intuitive Criterion attributes the deviation to the $H$-type and the profile fails, whereas for $V > \bar{V}(\rho)$ neither type is excluded on these grounds. Part (iii) restates cases (a) to (c); at the certification boundary itself the belief is unrestricted (Definition~\ref{def:d1-belief}).

\subsubsection*{The semi-separating equilibrium}

\begin{proposition}[Semi-separating equilibrium]\label{prop:semisep}
For $\bar{V}(\rho) \leq V \leq \bar{V}(\rho) + \underline{V}$, the profile in which $H$-type sellers choose Self and $L$-type sellers choose Self with probability $\sigma_L$ and None otherwise, with
\[
\sigma_L = \frac{\pi_0\,(1 - \pi_0 - y)}{(1 - \pi_0)(\pi_0 + y)}, \qquad y = \frac{\phi\rho}{V},
\]
and on-path belief $\hat{p}_{\text{Self}} = \pi_0 + \phi\rho/V$, is a PBE. The stated $\sigma_L$ lies in $[0, 1]$ throughout this range. For $\rho > 0$, the semi-separating profile is the unique equilibrium surviving the D1 criterion on the interior of the range.
\end{proposition}

\begin{proof}
\textit{Existence.} The $L$-type's indifference condition $\hat{p}_{\text{Self}} V - \phi\rho = \pi_0 V$ gives $\hat{p}_{\text{Self}} = \pi_0 + \phi\rho/V$, and Bayes' rule yields this belief at the stated $\sigma_L$ (the mixing probability derived in Methods). $\sigma_L \geq 0$ holds if and only if $y \leq 1 - \pi_0$, that is, $V \geq \bar{V}(\rho)$; $\sigma_L \leq 1$ reduces to $\phi\rho/V \geq 0$ and always holds. The $H$-type receives $\pi_0 V + \phi\rho$ from Self and $\pi_0 V$ from None. A deviation to TP pays $V - c_{\text{TP}}$ (Assumption 3), and
\[
\pi_0 V + \phi\rho \geq V - c_{\text{TP}} \iff V \leq \bar{V}(\rho) + \underline{V}.
\]
At $V = \bar{V}(\rho)$, $\sigma_L = 0$ and the profile coincides with the Self-separating equilibrium; at $V = \bar{V}(\rho) + \underline{V}$, $\sigma_L > 0$ and the $H$-type's payoffs from Self and TP are equal.

\textit{Uniqueness.} The remaining profiles fail on the interior of the range. (i) The Self-separating profile requires $V \leq \bar{V}(\rho)$ (Proposition~\ref{prop:self-existence}). (ii) The TP-separating profile fails the D1 criterion, by case (b). (iii) Pooling on Self gives $\hat{p}_{\text{Self}} = \pi_0$ by Bayes' rule, under which the $L$-type receives $\pi_0 V - \phi\rho < \pi_0 V$ and prefers None for $\rho > 0$. (iv) Pooling on None, with both types receiving $\pi_0 V$, is a PBE for $V < \underline{V}$ under off-path beliefs $\hat{p}_{\text{Self}}^{\text{off}} \leq \pi_0$. The deviation thresholds at Self are $\pi_0$ for the $H$-type and $\pi_0 + \phi\rho/V$ for the $L$-type, so for $\rho > 0$ the containment argument of Lemma~\ref{lem:threshold} sets the belief to $1$ and the $H$-type deviates. Weaker refinements do not remove this profile, so uniqueness on $\bar{V}(\rho) < V < \underline{V}$ relies on the D1 criterion. (v) Profiles in which the $H$-type mixes across signals, and profiles in which only the $L$-type sends Self, are excluded by Bayes' rule together with the $L$-type's participation condition.
\end{proof}

\begin{remark}[Robustness to Assumption 4]\label{rem:assumption3}
Assumption 4 enters the refinement only through the $L$-type's equilibrium payoff $\pi_0 V$. If None were instead subject to Bayesian updating, that payoff in the TP-separating profile would be $0$, the threshold $\bar{p}_L$ would become $\phi\rho/V$ with $\bar{p}_H$ unchanged, and case (b) would cover $\phi\rho < V < \phi\rho + c_{\text{TP}}$, an interval of width $c_{\text{TP}}$. The case structure persists.
\end{remark}

\begin{remark}[Price at the certification boundary]\label{rem:pricejump}
At $V = \bar{V}(\rho) + \underline{V}$ the $H$-type's payoff is continuous but the price is not. Below the boundary the item sells at $\hat{p}_{\text{Self}} V = \pi_0 V + \phi\rho$, which equals $V - c_{\text{TP}}$ on the boundary; above it the item is certified and sells at $V$. The price therefore jumps upward by exactly $c_{\text{TP}}$. The fee is passed through to buyers, and the seller's net payoff is unchanged across the boundary.
\end{remark}

\subsubsection*{Regime partition}

\begin{corollary}[Regime partition]\label{cor:regimes}
Propositions~\ref{prop:self-existence}, \ref{prop:tp-existence}, and~\ref{prop:semisep} together imply that the boundaries $\bar{V}(\rho)$ and $\bar{V}(\rho) + \underline{V}$ partition the $(\rho, V)$ space into three regimes:
\begin{itemize}
\item \textit{Self regime} ($V < \bar{V}(\rho)$): the Self-separating equilibrium $(H \to \text{Self},\ L \to \text{None})$;
\item \textit{Semi-separating regime} ($\bar{V}(\rho) < V < \bar{V}(\rho) + \underline{V}$): a band of vertical width $\underline{V}$ in which $H$-type sellers choose Self, $L$-type sellers mix between Self and None, and the belief is $\hat{p}_{\text{Self}} = \pi_0 + \phi\rho/V$;
\item \textit{TP regime} ($V > \bar{V}(\rho) + \underline{V}$): the TP-separating equilibrium $(H \to \text{TP},\ L \to \text{None})$.
\end{itemize}
The exclusion arguments (iii) to (v) in the proof of Proposition~\ref{prop:semisep} do not use the restriction of $V$ to the band, so for $\rho > 0$ the equilibrium is unique within each regime. The existence statements are to be read away from the certification boundary, where the $H$-type is indifferent between Self and TP (Remark~\ref{rem:pricejump}); the other boundary is covered by the propositions themselves.
\end{corollary}

\subsubsection*{Absence of certification}

The partition presumes that certification is available. For a seller who cannot obtain it, as in the reputation dynamics of Methods and Supplementary Note~8, the Semi-separating regime is not closed from above. The certification boundary arises solely from the $H$-type's deviation to TP in the existence proof of Proposition~\ref{prop:semisep}, a comparison that such a seller never faces, and the remaining exclusion arguments either do not involve TP or become vacuous without it. For $\rho > 0$, the semi-separating profile is therefore the unique equilibrium surviving the D1 criterion at every $V > \bar{V}(\rho)$.

\subsubsection*{Value-dependent certification cost}

The fee $c_{\text{TP}}$ is assumed independent of item value, whereas grading services charge more for cards of higher appraised value.

\begin{remark}[Value-dependent grading fee]\label{rem:valuefee}
Let the fee include a value-proportional component, $c_{\text{TP}}(V) = c_0 + \tau V$ with $c_0 > 0$ and $0 < \tau < 1 - \pi_0$, so that the proportional part does not absorb the entire quality premium. Substituting into the $H$-type's indifference condition replaces $\underline{V}$ by $c_0/(1 - \tau - \pi_0)$ and gives the certification boundary the slope $\phi/(1 - \tau - \pi_0)$, steeper than the slope $\phi/(1 - \pi_0)$ of $\bar{V}(\rho)$, which does not involve the fee and is unchanged. The band therefore widens with reputation instead of keeping a constant width, while the three regimes and the order of the two boundaries are unchanged. The main text adopts the fixed-fee approximation for parsimony.
\end{remark}

\subsection{Reputation Dynamics}\label{subsec:th-D}

Methods states a law of motion for the seller's reputation and the properties used in the main text. This note derives them, together with the finer structure of the certification switch. Consider the seller of Supplementary Note~7 for whom certification is unavailable, facing items of value $V$. Quality is drawn per listing, so a share $\pi_0$ of the seller's listings is high quality; on that share the $H$-type trades and reputation rises. On the remaining share the $L$-type self-claims with probability $\sigma_L(\rho)$, each misrepresentation lowers reputation by the same amount, and listings without a claim leave it unchanged. A single rate $\alpha > 0$ converts the net frequency of the two movements into a speed,
\begin{equation}
\dot{\rho} = \alpha\,[\,\pi_0 - (1 - \pi_0)\,\sigma_L(\rho)\,] ,
\label{eq:si-motion}
\end{equation}
with $\sigma_L$ as in Proposition~\ref{prop:semisep}. We call $\alpha(1 - \pi_0)\sigma_L(\rho)$ the erosion term and $\pi_0\alpha$ the full rate. The rate $\alpha$ fixes the unit of time and cancels from every rest point, existence condition, and comparison below. The penalty $\phi$ enters only through $y = \phi\rho/V$, the argument of $\sigma_L$. Reputation is bounded to $[0, 1]$. Where $\dot{\rho}$ is positive just below a point and negative just above it, trajectories from both sides converge to that point and remain there, and the same holds at an end of the interval toward which $\dot{\rho}$ points.

\begin{proposition}[Interior rest point]\label{prop:dyn-rest}
Equation~\ref{eq:si-motion} has at most one rest point in $(0, 1)$. One exists if and only if $\pi_0 < 1/2$ and $V(1/2 - \pi_0) < \phi$, in which case it lies at $\rho^{*} = V(1/2 - \pi_0)/\phi$ and is unstable.
\end{proposition}

\begin{proof}
\textit{Localization.} Outside the semi-separating range the $L$-type plays None, so $\sigma_L = 0$ and $\dot{\rho} = \pi_0\alpha > 0$. Every rest point therefore lies inside that range.

\textit{Monotonicity.} Differentiating $\sigma_L$ gives $d\sigma_L/dy = -\pi_0/[(1 - \pi_0)(\pi_0 + y)^2] < 0$, and $y = \phi\rho/V$ increases in $\rho$, so the right-hand side of Eq.~\ref{eq:si-motion} is strictly increasing in $\rho$ on the semi-separating range. It has at most one zero there and crosses it from below, so any rest point is unstable.

\textit{Location.} A rest point requires $\sigma_L = \pi_0/(1 - \pi_0)$, that is, $(1 - \pi_0 - y)/(\pi_0 + y) = 1$, so $y^{*} = 1/2 - \pi_0$ and $\rho^{*} = V y^{*}/\phi = V(1/2 - \pi_0)/\phi$.

\textit{Interiority.} $\rho^{*} > 0$ if and only if $\pi_0 < 1/2$, and $\rho^{*} < 1$ if and only if $V(1/2 - \pi_0) < \phi$. The semi-separating range ends at $y = 1 - \pi_0$, and $y^{*} = 1/2 - \pi_0 < 1 - \pi_0$ for every $\pi_0 \in (0, 1)$, so the localization step never excludes $\rho^{*}$.
\end{proof}

At the baseline values of Table~1 in the main text the second condition reads $V < \phi/(1/2 - \pi_0) = 2.5$, which every item value shown in Figs.~4 and~6 in the main text satisfies. Over that range the minority condition alone decides whether the rest point exists.

\begin{proposition}[Absorption at the floor]\label{prop:dyn-floor}
Under the conditions of Proposition~\ref{prop:dyn-rest}, $\dot{\rho} < 0$ throughout $[0, \rho^{*})$, so reputation starting below $\rho^{*}$ reaches $\rho = 0$ in finite time and remains there. The floor is not a rest point of Eq.~\ref{eq:si-motion}. At $\rho = 0$ a self-claim carries only the prior, so $\sigma_L = 1$ and $\dot{\rho} = \alpha(2\pi_0 - 1) < 0$. Reputation stops there because it is bounded below, not because the law of motion vanishes.
\end{proposition}

\begin{proof}
The right-hand side of Eq.~\ref{eq:si-motion} is strictly increasing on the semi-separating range and vanishes at $\rho^{*}$, so it is negative on $[0, \rho^{*})$ and bounded away from zero on any $[0, \rho_0]$ with $\rho_0 < \rho^{*}$, which gives the finite time. Setting $y = 0$ in $\sigma_L$ gives $\pi_0(1 - \pi_0)/[(1 - \pi_0)\pi_0] = 1$, whence the value at the floor.
\end{proof}

\begin{proposition}[The switch and the drop]\label{prop:dyn-switch}
Suppose certification is available. The $H$-type chooses TP below the switching reputation
\begin{equation}
\rho_s = [(1 - \pi_0)V - c_{\text{TP}}]/\phi ,
\label{eq:si-rho-s}
\end{equation}
and $\rho_s \in (0, 1)$ if and only if $\underline{V} < V < (\phi + c_{\text{TP}})/(1 - \pi_0)$. Below $\rho_s$ the $L$-type plays None and reputation grows at the full rate. At $\rho_s$ the growth rate falls by
\begin{equation}
\Delta = \frac{\pi_0\,\alpha\,c_{\text{TP}}}{V - c_{\text{TP}}} .
\label{eq:si-drop}
\end{equation}
\end{proposition}

\begin{proof}
\textit{Switch.} The certification boundary is $V = \bar{V}(\rho_s) + \underline{V} = [\phi\rho_s + c_{\text{TP}}]/(1 - \pi_0)$, which gives Eq.~\ref{eq:si-rho-s}. Then $\rho_s > 0$ if and only if $V > c_{\text{TP}}/(1 - \pi_0) = \underline{V}$, and $\rho_s < 1$ if and only if $(1 - \pi_0)V < \phi + c_{\text{TP}}$. Below the switch the static equilibrium is TP-separating, so $\sigma_L = 0$ and the erosion term vanishes.

\textit{Drop.} At the switch, $y_s = \phi\rho_s/V = (1 - \pi_0) - c_{\text{TP}}/V < 1 - \pi_0$, so just above $\rho_s$ the $L$-type mixes at the semi-separating rate. Then $1 - \pi_0 - y_s = c_{\text{TP}}/V$ and $\pi_0 + y_s = (V - c_{\text{TP}})/V$. Substituting into $\sigma_L$ gives $\sigma_L(\rho_s) = \pi_0 c_{\text{TP}}/[(1 - \pi_0)(V - c_{\text{TP}})]$, and the erosion term $\alpha(1 - \pi_0)\sigma_L(\rho_s)$ equals Eq.~\ref{eq:si-drop}.
\end{proof}

\begin{lemma}[The Semi-separating band]\label{lem:dyn-band}
The $L$-type mixes on the band $(\rho_s,\ \rho_s + c_{\text{TP}}/\phi)$ and plays None on either side of it. The erosion term is therefore positive throughout the band and zero outside it, and reputation grows at the full rate above the band as well as below the switch. A positive erosion term does not by itself make reputation fall; reputation falls only where the erosion term exceeds the full rate, a strict subinterval of the band whenever it is nonempty. Measured from the interior rest point, the two edges of the band lie at
\begin{equation*}
\rho_s - \rho^{*} = \frac{V - 2c_{\text{TP}}}{2\phi} ,
\qquad
\Big(\rho_s + \frac{c_{\text{TP}}}{\phi}\Big) - \rho^{*} = \frac{V}{2\phi} .
\end{equation*}
The upper distance is positive at every item value, so $\rho^{*}$ always lies below the top of the band. It lies above $\rho_s$ if and only if $V < 2c_{\text{TP}}$.
\end{lemma}

\begin{proof}
The band ends where $V = \bar{V}(\rho)$, at $\rho = (1 - \pi_0)V/\phi$. Subtracting $\rho^{*} = V(1/2 - \pi_0)/\phi$ from Eq.~\ref{eq:si-rho-s} gives $(V/2 - c_{\text{TP}})/\phi$, and subtracting it from the upper edge gives $V/(2\phi)$. The two distances differ by $c_{\text{TP}}/\phi$, the width of the band. The upper distance is positive for $V > 0$, and the lower is positive exactly when $V > 2c_{\text{TP}}$.
\end{proof}

\begin{proposition}[Trap at the switch]\label{prop:dyn-trap}
Suppose $\pi_0 < 1/2$ and a switch exists. If $V > 2c_{\text{TP}}$, reputation grows at every level and converges to $\rho = 1$. If $V < 2c_{\text{TP}}$ and $V(1/2 - \pi_0) < \phi$, reputation below $\rho^{*}$ converges to $\rho_s$ and reputation above $\rho^{*}$ converges to $\rho = 1$. If $V < 2c_{\text{TP}}$ and $V(1/2 - \pi_0) > \phi$, reputation converges to $\rho_s$ from every starting point. In the last two cases the stall is at the discontinuity $\rho_s$, not at a zero of Eq.~\ref{eq:si-motion}, because the growth rate is the full rate just below the switch and negative just above it. At the equalities $V = 2c_{\text{TP}}$ and $V(1/2 - \pi_0) = \phi$, the rest point coincides with the switch or with $\rho = 1$, respectively, and the statements hold for every starting point other than that point.
\end{proposition}

\begin{proof}
\textit{Signs.} Below the switch and above the band the rate is the full rate, by Proposition~\ref{prop:dyn-switch} and Lemma~\ref{lem:dyn-band}. Inside the band the right-hand side of Eq.~\ref{eq:si-motion} is strictly increasing with its unique zero at $\rho^{*}$, since the monotonicity step of Proposition~\ref{prop:dyn-rest} applies unchanged.

\textit{Cases.} If $V > 2c_{\text{TP}}$, Lemma~\ref{lem:dyn-band} places that zero below the band, so reputation grows at every level. If $V < 2c_{\text{TP}}$, the same lemma places it above the switch. The right-hand side is then negative exactly on $(\rho_s, \rho^{*})$. Reputation there falls back to $\rho_s$, and reputation below $\rho_s$ rises to it at the full rate.

\textit{Reachability.} When $\rho^{*} < 1$, reputation above it leaves the band and converges to $\rho = 1$. When $\rho^{*} > 1$, the lemma places the top of the band above $\rho^{*}$ and hence above one, so $(\rho_s, 1]$ lies inside the band and the right-hand side is negative throughout. Reputation then converges to $\rho_s$ from every starting point.

Equivalently, the growth rate just above the switch is $\pi_0\alpha - \Delta = \pi_0\alpha\,(V - 2c_{\text{TP}})/(V - c_{\text{TP}})$, so the drop stalls reputation exactly when $V < 2c_{\text{TP}}$.
\end{proof}

\subsubsection*{The continuous approximation}

The dynamics treat the reputation record as a continuous variable evolving deterministically, whereas the record advances one transaction at a time. The differential equation is therefore an approximation that improves with the number of recorded transactions. Sellers in our data post about 16 listings on average, so Fig.~6 in the main text is a qualitative illustration of the equilibrium structure rather than a calibrated trajectory.

\subsection{Numerical Illustration}\label{subsec:th-E}

Supplementary Notes~6 and 7 give $\bar{V}(\rho)$ (Eq.~\ref{eq:bar-V}) and $\underline{V}$ (Eq.~\ref{eq:underline-V}) explicitly. For the qualitative illustration in Fig.~4 in the main text we use the reference values $\pi_0 = 0.3$, $\phi = 0.5$, and $c_{\text{TP}} = 0.2$, which give $\bar{V}(1) = 0.714$, $\underline{V} = 0.286$, and a certification boundary of $\bar{V}(1) + \underline{V} = 1.000$ at full reputation. The reputation dynamics of Supplementary Note~8 add the rate $\alpha > 0$, which only fixes the unit of time and therefore has no row in Table~1 in the main text. Fig.~6 in the main text is drawn at $\alpha = 0.1$, which sets the vertical scale of the phase line.

\label{sec:si-methods-end}
\FloatBarrier
\clearpage

\SIresults

\section*{Descriptive Price Analysis}\label{sec:res-A}

Here we show that both pricing patterns hold in descriptive regressions of the log price ratio. Reputation raises the price premium from self-grading but lowers the premium from TP grading. We document these patterns with linear mixed-effects models with seller random effects, regressing the log price ratio on reputation, signal type, and their interaction while controlling for market price bracket and condition category (Supplementary Fig.~\ref{fig:S1}). Across all ungraded listings, the predicted log price ratio rises with reputation under self-grading (interaction slope $\beta = 0.073$; panel A), whereas the corresponding slope for TP grading is negative ($\beta = -0.118$; panel C), so high-reputation sellers obtain a smaller incremental premium from TP grading. Both patterns persist among sold items ($\beta = 0.034$, SE $= 0.006$, $p < 0.001$ for self-grading, panel B; $\beta = -0.179$ for TP grading, panel D), so selection from unsold listings does not drive them.

\begin{figure}[H]
\centering
\makebox[\textwidth][c]{\includegraphics[width=18cm]{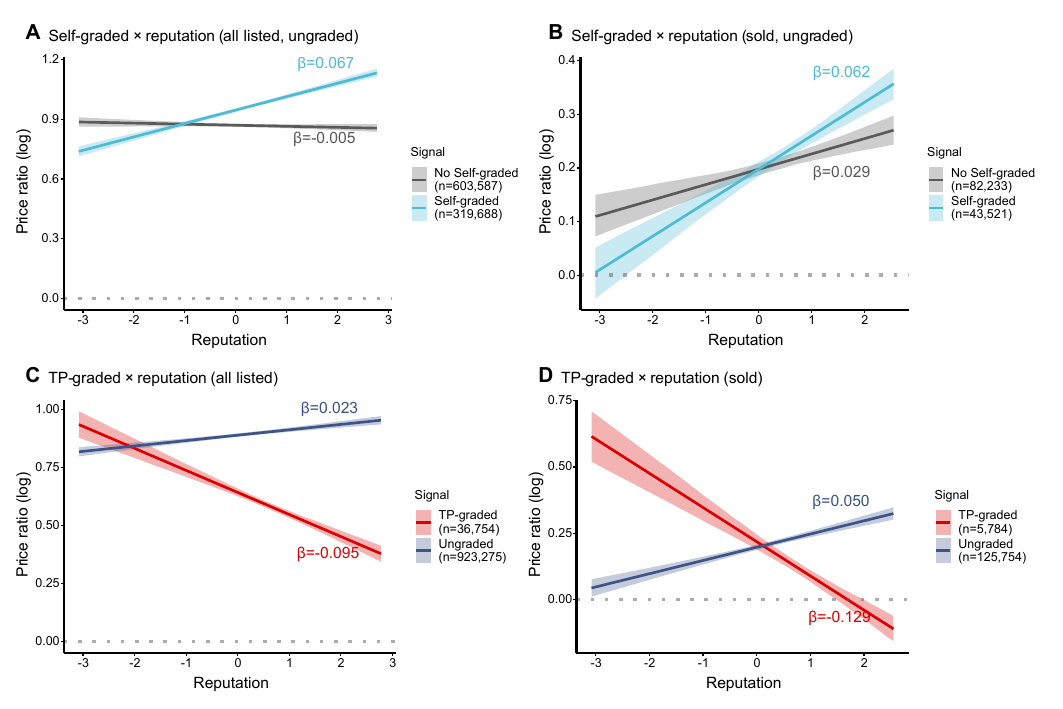}}
\caption{Descriptive interaction plots of seller reputation and signal type. Each panel shows predictions and prediction intervals from a linear mixed-effects model with seller random effects, controlling for market price bracket and condition category. \textbf{(A)} Reputation $\times$ self-graded, all ungraded listings. \textbf{(B)} Reputation $\times$ self-graded, sold ungraded items only ($N = 125{,}754$). \textbf{(C)} Reputation $\times$ TP-graded, all listings. \textbf{(D)} Reputation $\times$ TP-graded, sold items only.}
\label{fig:S1}
\end{figure}

\clearpage
\section*{Seller Volume}\label{sec:res-B}

Here we show that signal choice is associated with transaction volume as well as per-unit prices. Using data aggregated at the seller-by-signal level (sellers with at least three listings per signal), we compare three volume metrics across reputation levels (Supplementary Fig.~\ref{fig:S2}). High-reputation sellers list more items under every signal, although the increase is modest for TP-graded items (panel A), and their TP-graded listings sell at particularly high rates (panel B). The two effects combine in sale counts, so high-reputation sellers realize greater transaction volumes under every signal (panel C).

\begin{figure}[H]
\centering
\makebox[\textwidth][c]{\includegraphics[width=18cm]{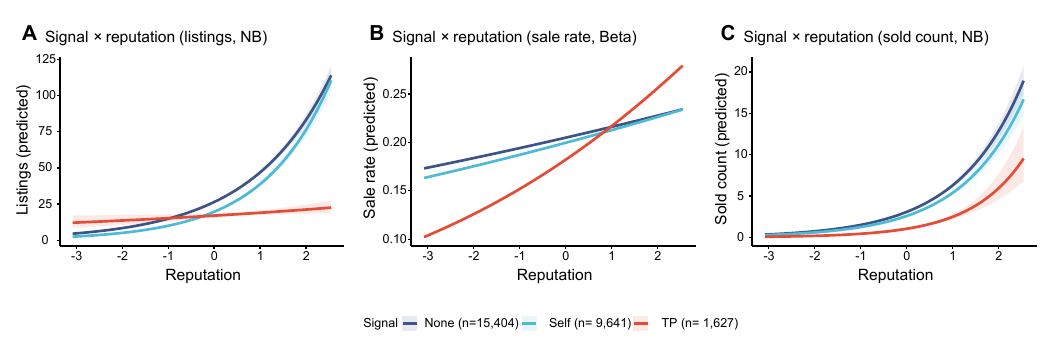}}
\caption{Seller-level volume metrics by signal type (None, Self-graded, TP-graded). \textbf{(A)} Listing count, negative binomial regression. \textbf{(B)} Sale rate, beta regression for proportions in $(0, 1)$. \textbf{(C)} Sale count, negative binomial regression.}
\label{fig:S2}
\end{figure}

\clearpage
\section*{Revenue and Margin Analysis}\label{sec:res-C}

Supplementary Fig.~\ref{fig:S3} shows revenue and margin by signal type across the reputation range, conditional on sale (panels A and B, from linear mixed-effects models with seller random effects) and in expectation (panels C and D, from a two-part model that multiplies the sale probability from a logistic mixed-effects model by the conditional expectation from the linear model). TP-graded items command higher conditional prices and margins, but the advantage erodes at high reputation (panels A and B). High-reputation sellers of TP-graded items accept smaller conditional margins, consistent with a volume-oriented strategy, and their higher sale probabilities (Supplementary Fig.~\ref{fig:S2}) compensate. Expected revenue for TP-graded items rises steeply with reputation (panel C), and expected margin follows an inverted-U pattern across reputation levels (panel D).

\begin{figure}[H]
\centering
\makebox[\textwidth][c]{\includegraphics[width=18cm]{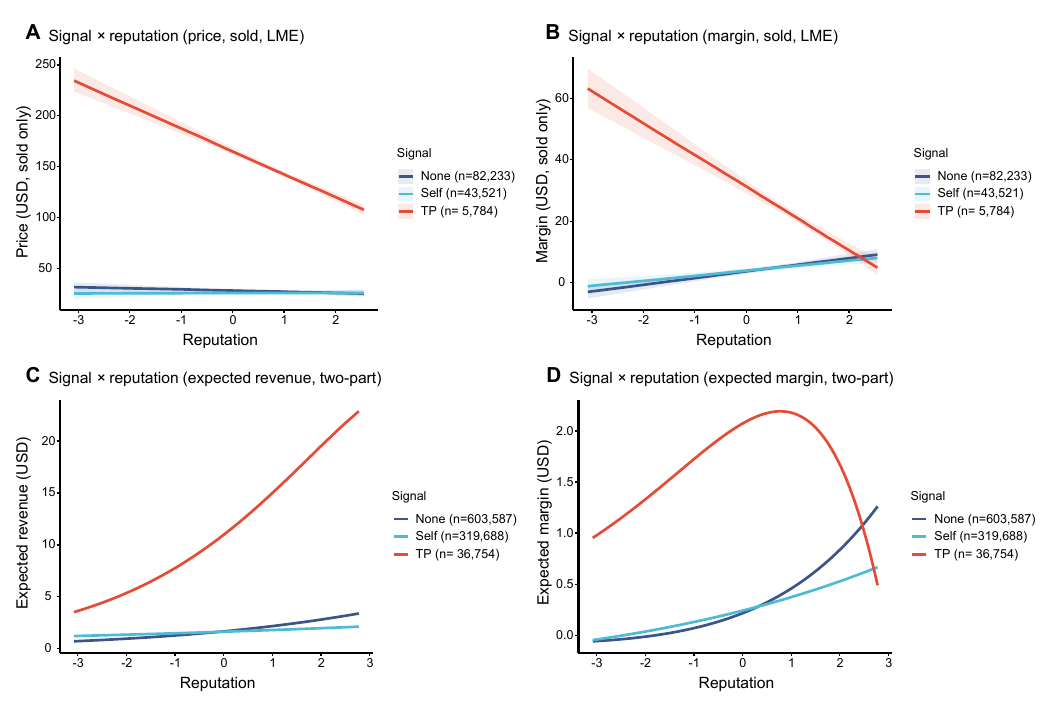}}
\caption{Revenue and margin by signal type and seller reputation, conditional and expected. \textbf{(A)} Unit sale price (USD) conditional on sale, by signal type (None, Self-graded, TP-graded), from a linear mixed-effects model with seller random effects. \textbf{(B)} Margin (sale price minus market reference price) conditional on sale, from the same model. \textbf{(C)} Expected revenue, $P(\text{sold}) \times E[\text{price} \mid \text{sold}]$, from a two-part model combining a logistic mixed-effects model (sale probability) and the linear model (conditional price). \textbf{(D)} Expected margin, $P(\text{sold}) \times E[\text{margin} \mid \text{sold}]$, from the same two-part model.}
\label{fig:S3}
\end{figure}

\clearpage
\section*{Continuous Regime Boundary in Reputation--Value Space}\label{sec:res-E}

Here we show that the regime partition reported in the main text has a continuous empirical counterpart. The empirical region map (Fig.~2B in the main text), the three-regime map (Fig.~4), and the heatmap of TP shares (Fig.~5) summarize the same partition at different levels of aggregation, and Supplementary Fig.~\ref{fig:S4} shows the continuous boundary underlying them. Panel~A is a scatter plot of self-graded (blue) and TP-graded (red) listings in the space of the seller reputation score and the log ungraded market price, overlaid with an iso-probability line of a logistic regression of signal choice on these two variables. The logistic regression gives $P(\text{TP} \mid \text{reputation}, \text{price})$, and the line is drawn where this probability equals the TP base rate of $10.3\%$, the share of TP-graded listings among graded listings. Bayes' rule then implies that along the line the density of TP-graded listings equals that of self-graded listings. The line puts $74.5\%$ of TP-graded and $21.5\%$ of self-graded listings on its TP side. Panel~B bins all graded listings in the same coordinates and shows, separately for each signal, where its listings concentrate. The mass of TP-graded listings sits above the line and that of self-graded listings below it. The iso-probability lines are parallel, so one unit of reputation raises the price threshold for TP grading by a factor of $\exp(0.24) \approx 1.27$. This descriptive slope agrees with the HBM signal-selection coefficients, which imply $0.236$ at mean reputation (Supplementary Table~\ref{tab:coef}) and a range of $0.210$ to $0.264$ across the observed reputation range.

\begin{figure}[H]
\centering
\makebox[\textwidth][c]{\includegraphics[width=18cm]{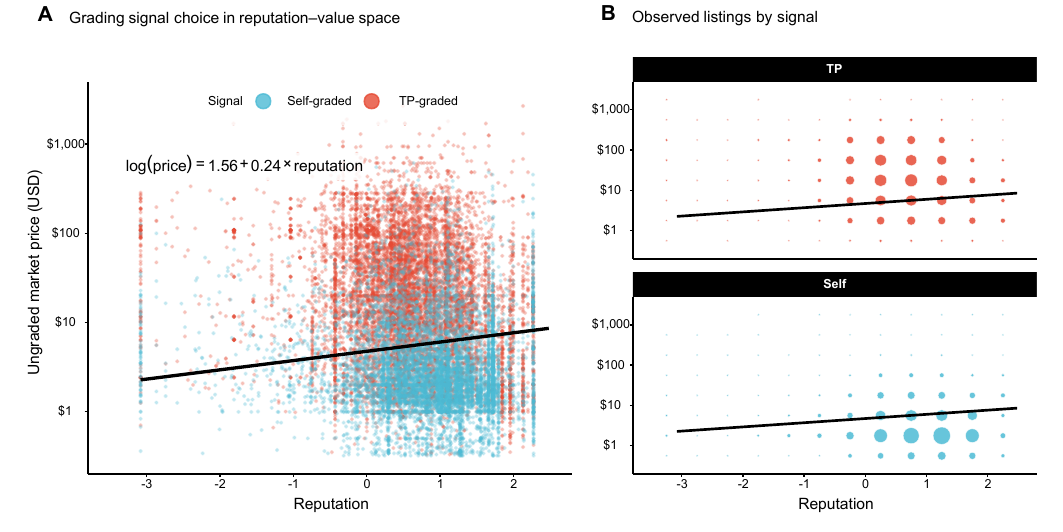}}
\caption{Continuous regime boundary in reputation--value space. \textbf{(A)} Scatter plot of self-graded (blue; $n = 10{,}000$, randomly sampled from the analysis sample) and TP-graded (red; $n = 10{,}000$, randomly sampled) listings in the space of the seller reputation score and the ungraded market price (log scale). The black line is the equal-density boundary from a logistic regression of signal choice on the log ungraded market price and the seller reputation score, fitted to all graded listings: $\log(\text{market price}) = 1.56 + 0.24 \times (\text{reputation score})$, along which $P(\text{TP} \mid \text{graded})$ equals the TP base rate ($10.3\%$) and the two signals' densities are equal. \textbf{(B)} All TP-graded (top, red; $n = 36{,}754$) and self-graded (bottom, blue; $n = 319{,}688$) listings in the same space, binned at 0.5 reputation units by half a decade in price. Bubble area is proportional to the share of that signal's listings falling in the bin, so the two panels are comparable in shape but not in absolute count. The black line repeats the boundary from (A).}
\label{fig:S4}
\end{figure}

\clearpage
\section*{MCMC Convergence Diagnostics}\label{sec:res-F}

Supplementary Table~\ref{tab:convergence} lists the potential scale reduction factor $\hat{R}$ and the bulk and tail effective sample sizes (ESS) for the 38 reported quantities of the hierarchical Bayesian model, namely the 36 model parameters and the two intraclass correlations derived from them. $\hat{R}$ ranges from $1.000$ to $1.013$ and bulk ESS from $221$ to $3{,}883$. Both extremes occur for the random-effects sold--signal correlation $\Omega_{23}$, which is not among the quantities reported in the main text. The eight chains therefore mixed and explored the posterior adequately.

\begin{table}[H]
\centering
\caption{Convergence diagnostics for the 38 reported quantities of the hierarchical Bayesian model.}\label{tab:convergence}
\fontsize{7}{8.4}\selectfont
\begin{tabular}[t]{lccc}
\toprule
\fontsize{8}{9.6}\selectfont Parameter & $\hat{R}$ & ESS (bulk) & ESS (tail)\\
\midrule
\addlinespace[0.3em]
\multicolumn{4}{l}{\fontsize{8}{9.6}\selectfont \textbf{Intercepts}}\\
\hspace{1em}Price Intercept ($\alpha_{\text{price}}$) & 1.002 & 1126 & 2131\\
\hspace{1em}Sold Intercept ($\alpha_{\text{sold}}$) & 1.001 & 3132 & 3672\\
\addlinespace[0.3em]
\multicolumn{4}{l}{\fontsize{8}{9.6}\selectfont \textbf{Signal Selection}}\\
\hspace{1em}Intercept ($\alpha$), Self-graded & 1.008 & 1238 & 1973\\
\hspace{1em}Intercept ($\alpha$), TP-graded & 1.002 & 1928 & 2155\\
\hspace{1em}Reputation ($\gamma$), Self-graded & 1.006 & 1525 & 2644\\
\hspace{1em}Reputation ($\gamma$), TP-graded & 1.003 & 2022 & 2764\\
\hspace{1em}Market Price ($\delta$), Self-graded & 1.000 & 3382 & 3536\\
\hspace{1em}Market Price ($\delta$), TP-graded & 1.000 & 3267 & 3472\\
\hspace{1em}Reputation $\times$ Market Price ($\psi$), Self-graded & 1.002 & 3479 & 3771\\
\hspace{1em}Reputation $\times$ Market Price ($\psi$), TP-graded & 1.000 & 3433 & 3653\\
\hspace{1em}Japanese Card ($\zeta$), Self-graded & 1.002 & 3649 & 3677\\
\hspace{1em}Japanese Card ($\zeta$), TP-graded & 1.000 & 3381 & 3608\\
\addlinespace[0.3em]
\multicolumn{4}{l}{\fontsize{8}{9.6}\selectfont \textbf{Price Equation}}\\
\hspace{1em}Reputation ($\beta_1$) & 1.002 & 934 & 1747\\
\hspace{1em}Japanese Card ($\beta_2$) & 1.000 & 3397 & 3326\\
\hspace{1em}Reputation $\times$ Self-graded ($\beta_3$) & 1.002 & 3645 & 3752\\
\hspace{1em}Reputation $\times$ TP-graded ($\beta_4$) & 1.002 & 3883 & 3861\\
\hspace{1em}Returns Accepted ($\beta_5$) & 1.001 & 2797 & 3359\\
\hspace{1em}Self-graded ($\beta_{\text{signal}}$) & 1.000 & 3372 & 3446\\
\hspace{1em}TP-graded ($\beta_{\text{signal}}$) & 1.001 & 3092 & 3670\\
\addlinespace[0.3em]
\multicolumn{4}{l}{\fontsize{8}{9.6}\selectfont \textbf{Sold Equation}}\\
\hspace{1em}Reputation ($\beta_1$) & 1.001 & 1715 & 2827\\
\hspace{1em}Self-graded ($\beta_2$) & 1.004 & 1330 & 2855\\
\hspace{1em}TP-graded ($\beta_3$) & 1.002 & 2570 & 3285\\
\hspace{1em}Japanese Card ($\beta_4$) & 1.001 & 3542 & 3576\\
\hspace{1em}Returns Accepted ($\beta_5$) & 1.001 & 2736 & 3543\\
\hspace{1em}Price Ratio (log) ($\beta_6$) & 1.000 & 3208 & 3347\\
\hspace{1em}Listing Duration (log) ($\beta_7$) & 1.001 & 3460 & 3469\\
\hspace{1em}Reputation $\times$ Self-graded ($\beta_8$) & 1.001 & 2525 & 3481\\
\hspace{1em}Reputation $\times$ TP-graded ($\beta_9$) & 1.001 & 3446 & 3233\\
\addlinespace[0.3em]
\multicolumn{4}{l}{\fontsize{8}{9.6}\selectfont \textbf{Variance/Structural}}\\
\hspace{1em}Price Residual SD ($\sigma_{\text{price}}$) & 1.002 & 3872 & 3388\\
\hspace{1em}Price RE SD ($\sigma_{u_1}$) & 1.012 & 798 & 1617\\
\hspace{1em}Sold RE SD ($\sigma_{u_2}$) & 1.005 & 1172 & 2093\\
\hspace{1em}Signal RE SD ($\sigma_{u_3}$) & 1.002 & 1287 & 2182\\
\hspace{1em}Price ICC & 1.012 & 831 & 1579\\
\hspace{1em}Sold ICC & 1.005 & 1172 & 2093\\
\hspace{1em}TP Factor Loading ($\lambda^{\text{TP}}$) & 1.000 & 3020 & 3513\\
\hspace{1em}RE Corr: Price--Sold ($\Omega_{12}$) & 1.002 & 1383 & 2248\\
\hspace{1em}RE Corr: Price--Signal ($\Omega_{13}$) & 1.005 & 822 & 1459\\
\hspace{1em}RE Corr: Sold--Signal ($\Omega_{23}$) & 1.013 & 221 & 532\\
\bottomrule
\end{tabular}

\end{table}

\clearpage
\section*{HBM Model Fit}\label{sec:res-G}

The predictive fit of the three-equation model complements the convergence diagnostics. The model performs adequately in each equation, with Cohen's $\kappa = 0.584$ for signal choice, $R^2 = 0.421$ for the log price ratio, and AUC $= 0.954$ for the sale outcome. Supplementary Fig.~\ref{fig:S5} shows that posterior predictive accuracy exceeds the empirical base rate for all three signal categories (panel A), that the confusion matrix is dominated by its diagonal with most off-diagonal mass between None and Self-graded (panel B), that the observed-versus-predicted scatter for the log price ratio shows slightly wider residuals at high-price TP-graded listings (panel C), that the calibration slope for sale probability is close to one (panel D), and that the ROC curve confirms strong discrimination for the sale outcome (panel E).

\begin{figure}[H]
\centering
\makebox[\textwidth][c]{\includegraphics[width=18cm]{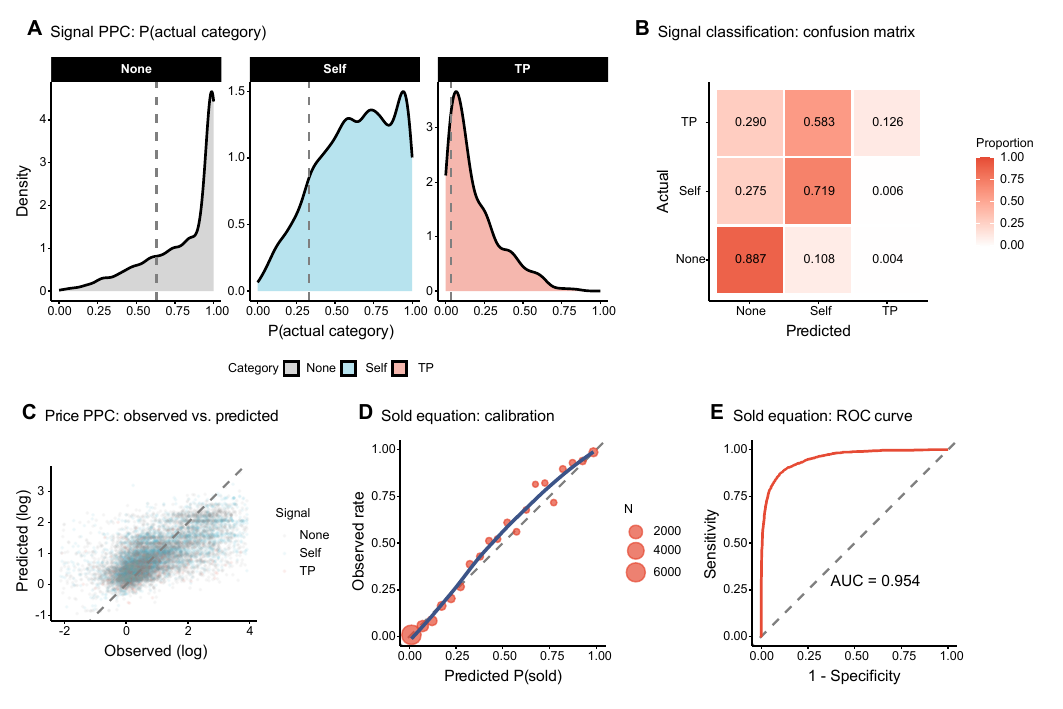}}
\caption{Model fit and predictive diagnostics across the three equations. \textbf{(A)} Posterior predictive check for signal choice, the density of $P(\text{correct category})$ by signal type; dashed lines indicate empirical base rates. \textbf{(B)} Confusion matrix for signal choice; overall accuracy $= 0.803$, balanced accuracy $= 0.577$. \textbf{(C)} Observed versus predicted log price ratio; RMSE $= 0.820$, normalized RMSE $= 0.762$. Residual distributions are similar across signal types. \textbf{(D)} Calibration plot for sale predictions, binned mean predicted $P(\text{sold})$ versus observed sale rate; calibration slope $= 1.111$. \textbf{(E)} ROC curve for sale predictions; 95\% CI for the AUC, $[0.948, 0.960]$.}
\label{fig:S5}
\end{figure}

\clearpage
\section*{Full Coefficient Estimates}\label{sec:res-H}

Supplementary Table~\ref{tab:coef} reports the posterior estimates for the 38 reported quantities of the hierarchical Bayesian model, the 36 parameters and the two intraclass correlations, with the price-equation interaction coefficients forming the structural core of the main-text argument. Estimates are posterior medians with 95\% credible intervals (2.5th and 97.5th percentiles).

\begin{table}[H]
\centering
\caption{Full posterior estimates for the 38 reported quantities of the hierarchical Bayesian model.}\label{tab:coef}
\fontsize{7}{8.4}\selectfont
\begin{tabular}[t]{lc}
\toprule
\fontsize{8}{9.6}\selectfont Parameter & Estimate [95\% CI]\\
\midrule
\addlinespace[0.3em]
\multicolumn{2}{l}{\fontsize{8}{9.6}\selectfont \textbf{Intercepts}}\\
\hspace{1em}Price Intercept ($\alpha_{\text{price}}$) & 0.853 [0.846, 0.861]\\
\hspace{1em}Sold Intercept ($\alpha_{\text{sold}}$) & 5.091 [5.032, 5.148]\\
\addlinespace[0.3em]
\multicolumn{2}{l}{\fontsize{8}{9.6}\selectfont \textbf{Signal Selection}}\\
\hspace{1em}Intercept ($\alpha$), Self-graded & -0.497 [-0.524, -0.469]\\
\hspace{1em}Intercept ($\alpha$), TP-graded & -4.140 [-4.180, -4.102]\\
\hspace{1em}Reputation ($\gamma$), Self-graded & 0.054 [0.031, 0.077]\\
\hspace{1em}Reputation ($\gamma$), TP-graded & -0.052 [-0.083, -0.020]\\
\hspace{1em}Market Price ($\delta$), Self-graded & -0.252 [-0.259, -0.244]\\
\hspace{1em}Market Price ($\delta$), TP-graded & 0.505 [0.495, 0.515]\\
\hspace{1em}Reputation $\times$ Market Price ($\psi$), Self-graded & 0.044 [0.037, 0.049]\\
\hspace{1em}Reputation $\times$ Market Price ($\psi$), TP-graded & 0.028 [0.020, 0.037]\\
\hspace{1em}Japanese Card ($\zeta$), Self-graded & 0.087 [0.065, 0.107]\\
\hspace{1em}Japanese Card ($\zeta$), TP-graded & 1.008 [0.979, 1.038]\\
\addlinespace[0.3em]
\multicolumn{2}{l}{\fontsize{8}{9.6}\selectfont \textbf{Price Equation}}\\
\hspace{1em}Reputation ($\beta_1$) & 0.001 [-0.006, 0.007]\\
\hspace{1em}Japanese Card ($\beta_2$) & -0.503 [-0.510, -0.497]\\
\hspace{1em}Reputation $\times$ Self-graded ($\beta_3$) & 0.063 [0.058, 0.069]\\
\hspace{1em}Reputation $\times$ TP-graded ($\beta_4$) & -0.095 [-0.110, -0.080]\\
\hspace{1em}Returns Accepted ($\beta_5$) & 0.026 [0.020, 0.032]\\
\hspace{1em}Self-graded ($\beta_{\text{signal}}$) & 0.093 [0.086, 0.100]\\
\hspace{1em}TP-graded ($\beta_{\text{signal}}$) & -0.197 [-0.213, -0.180]\\
\addlinespace[0.3em]
\multicolumn{2}{l}{\fontsize{8}{9.6}\selectfont \textbf{Sold Equation}}\\
\hspace{1em}Reputation ($\beta_1$) & 0.457 [0.427, 0.488]\\
\hspace{1em}Self-graded ($\beta_2$) & 0.102 [0.065, 0.138]\\
\hspace{1em}TP-graded ($\beta_3$) & 0.026 [-0.062, 0.115]\\
\hspace{1em}Japanese Card ($\beta_4$) & -0.380 [-0.414, -0.348]\\
\hspace{1em}Returns Accepted ($\beta_5$) & -0.256 [-0.289, -0.220]\\
\hspace{1em}Price Ratio (log) ($\beta_6$) & -0.951 [-0.964, -0.939]\\
\hspace{1em}Listing Duration (log) ($\beta_7$) & -1.654 [-1.666, -1.643]\\
\hspace{1em}Reputation $\times$ Self-graded ($\beta_8$) & -0.077 [-0.109, -0.045]\\
\hspace{1em}Reputation $\times$ TP-graded ($\beta_9$) & 0.094 [0.018, 0.170]\\
\addlinespace[0.3em]
\multicolumn{2}{l}{\fontsize{8}{9.6}\selectfont \textbf{Variance/Structural}}\\
\hspace{1em}Price Residual SD ($\sigma_{\text{price}}$) & 0.838 [0.837, 0.840]\\
\hspace{1em}Price RE SD ($\sigma_{u_1}$) & 0.638 [0.632, 0.645]\\
\hspace{1em}Sold RE SD ($\sigma_{u_2}$) & 1.847 [1.816, 1.877]\\
\hspace{1em}Signal RE SD ($\sigma_{u_3}$) & 2.110 [2.081, 2.139]\\
\hspace{1em}Price ICC & 0.367 [0.362, 0.372]\\
\hspace{1em}Sold ICC & 0.509 [0.501, 0.517]\\
\hspace{1em}TP Factor Loading ($\lambda^{\text{TP}}$) & 0.919 [0.911, 0.926]\\
\hspace{1em}RE Corr: Price--Sold ($\Omega_{12}$) & -0.015 [-0.036, 0.004]\\
\hspace{1em}RE Corr: Price--Signal ($\Omega_{13}$) & -0.139 [-0.153, -0.125]\\
\hspace{1em}RE Corr: Sold--Signal ($\Omega_{23}$) & 0.019 [0.001, 0.038]\\
\bottomrule
\end{tabular}

\end{table}

\clearpage
\section*{Signal Composition Across Price and Reputation Dimensions}\label{sec:res-I}

Here we show that signal composition varies differently along the price and reputation dimensions. Supplementary Fig.~\ref{fig:S6} disaggregates the $5 \times 5$ heatmap of Fig.~5 in the main text into marginal distributions by price quintile and by reputation quintile. The price distribution is unimodal, peaking at \$1 to \$3 with a long right tail (panel A), and the reputation distribution is bimodal, with a low-reputation cluster at scores of $-1$ to $0$ and a high-reputation cluster at $0.5$ to $1.5$ (panel C). Along the price axis (panel B), the TP share rises monotonically from $0.3\%$ in the lowest quintile to $12.2\%$ in the highest, the Self share declines correspondingly, and the None share is nearly unchanged. Along the reputation axis (panel D), the Self share rises monotonically from $28.7\%$ to $35.4\%$, the TP share declines from $5.6\%$ to $2.2\%$, and the None share is again nearly unchanged. Variation in signal use thus occurs between Self and TP rather than in whether a signal is used at all. The TP share varies mainly with price, whereas the Self share rises with reputation, in line with the two-dimensional regime structure.

\begin{figure}[H]
\centering
\makebox[\textwidth][c]{\includegraphics[width=18cm]{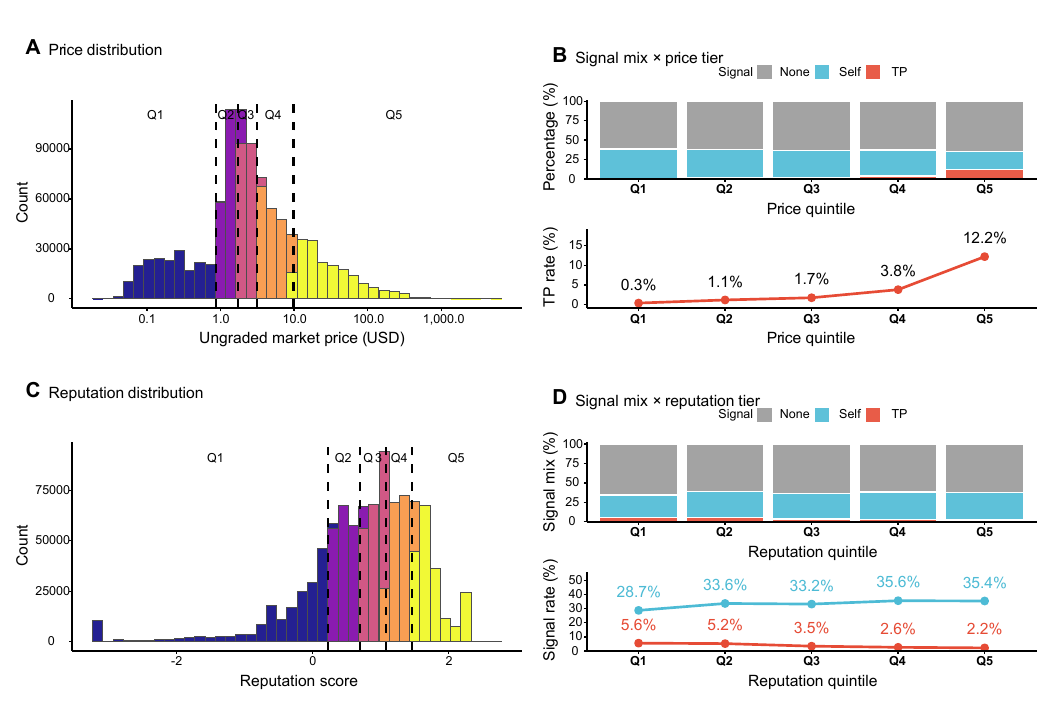}}
\caption{Empirical distributions and signal composition across price and reputation. \textbf{(A)} Distribution of ungraded market prices ($N = 960{,}029$), a log-scale histogram of listing counts with quintile boundaries. \textbf{(B)} Signal composition by market price quintile, stacked to 100\%. \textbf{(C)} Distribution of seller reputation scores, a histogram of listing counts with quintile boundaries. \textbf{(D)} Signal composition by reputation quintile.}
\label{fig:S6}
\end{figure}

\label{sec:si-results-end}
\clearpage
\FloatBarrier

\label{sec:si-refs}

\end{document}